\documentclass[11pt]{article}
\usepackage[a4paper, margin=1.5in]{geometry}
\usepackage{amsthm,amsmath,amssymb,amsfonts, amsbsy, mathtools, epsfig}
\usepackage[caption=false]{subfig}
\usepackage[usenames]{color}
\usepackage[figuresright]{rotating}
\usepackage{boxedminipage}
\usepackage[colorlinks,citecolor=blue,urlcolor=blue]{hyperref}
\usepackage{epstopdf}
\usepackage{natbib}
\usepackage{multirow}
\usepackage{enumerate}
\usepackage{verbatim}
\usepackage{bbm}
\usepackage{xcolor}
\usepackage{bm}
\usepackage{mathrsfs,color,dsfont}
\usepackage{algorithm}
\usepackage{algorithmicx}
\usepackage{algpseudocode}
\newcommand{\ignore}[1]{}{}

\newtheorem{theorem}{Theorem}
\newtheorem{proposition}{Proposition}
\newtheorem{lemma}{Lemma}
\newtheorem{remark}{Remark}
\newcommand{\RR}{\R}

\newcommand{\ep}{\epsilon}

\newcommand{\argmin}{\mathop{\rm arg\min}}

\def\0{\boldsymbol{0}}

\def\b{\boldsymbol{b}}
\def\x{\boldsymbol{x}}

\def\R{\mathbb{R}}

\def\D{\boldsymbol{D}}
\def\hD{\widehat{\boldsymbol{D}}}
\def\A{\boldsymbol{A}}
\def\z{\boldsymbol{z}}

\def\H{\boldsymbol{H}}

\def\x{\boldsymbol{x}}
\def\y{\boldsymbol{y}}

\def\be{\boldsymbol{\beta}}

\def\X{\boldsymbol{X}}

\def\S{\boldsymbol{\Sigma}}

\def\0{\boldsymbol{0}}

\def\x{\boldsymbol{x}}

\def\be{\boldsymbol{\beta}}

\def\X{\boldsymbol{X}}

\def\S{\boldsymbol{\Sigma}}

\def\figspace{0.39}

\def\pr{\mathbb{P}} 
\def\ep{\mathbb{E}} 

\newcommand{\hbe}{\widehat{\be}}
\newcommand{\hS}{\widehat{\S}}

\newcommand{\mbE}{\mathbb{E}}

\newcommand{\mbR}{\mathbb{R}}
\newcommand{\mbP}{\mathbb{P}}

\newcommand{\tY}{\widetilde{Y}}
\newcommand{\tX}{\widetilde{\boldsymbol{X}}}
\newcommand{\bX}{\bm{X}}

\newcommand{\bA}{\bm{A}}
\newcommand{\bC}{\bm{C}}

\newcommand{\balpha}{\bm{\alpha}}
\newcommand{\bbeta}{\bm{\beta}}
\newcommand{\bhbeta}{\widehat{\bbeta}}
\newcommand{\btbeta}{\widetilde{\bbeta}}

\newcommand{\bZ}{\bm{Z}}
\newcommand{\btZ}{\widetilde{\bZ}}

\newcommand{\bG}{\bm{G}}

\newcommand{\btheta}{\bm{\theta}}

\newcommand{\bv}{\bm{v}}
\newcommand{\bSigma}{\bm{\Sigma}}
\newcommand{\bhSigma}{\widehat{\bSigma}}

\newcommand{\norm}[1]{\Vert#1\Vert}
\newcommand{\Norm}[1]{\left\Vert#1\right\Vert}

\newcommand{\abs}[1]{\vert#1\vert}
\newcommand{\Abs}[1]{\left\vert#1\right\vert}
\newcommand{\ind}[1]{\mathds{1}[#1]}
\newcommand{\Ind}[1]{\mathds{1}\left[#1\right]}
\let\tilde\widetilde

\makeatletter
\newtheorem{rep@theorem}{\rep@title}
\newcommand{\newreptheorem}[2]{%
	\newenvironment{rep#1}[1]{%
		\def\rep@title{#2 \ref{##1}}%
		\begin{rep@theorem}}%
		{\end{rep@theorem}}}
\makeatother
\newreptheorem{theorem}{Theorem}
\newreptheorem{proposition}{Proposition}
\newreptheorem{lemma}{Lemma}

\definecolor{DSgray}{cmyk}{0,1,0,0}

\author{Xi Chen
	\footnote{New York University, New York, USA, Email: xchen3@stern.nyu.edu} ~
	Weidong Liu
	\footnote{Shanghai Jiao Tong University, Shanghai, China, Email: weidongl@sjtu.edu.cn} ~
	Xiaojun Mao
	\footnote{Fundan University,
		Shanghai,
		China,
		Email: maoxj@fudan.edu.cn} ~
	Zhuoyi Yang
	\footnote{New York University,
		New York,
		USA,
		Email: zyang@stern.nyu.edu}}

\begin{document}
	\title{Distributed High-dimensional  Regression Under a Quantile Loss Function}
	\author{Xi Chen
		\footnote{New York University, New York, USA, Email: xchen3@stern.nyu.edu} ~
		Weidong Liu
		\footnote{Shanghai Jiao Tong University, Shanghai, China, Email: weidongl@sjtu.edu.cn} ~
		Xiaojun Mao
		\footnote{Fundan University,
			Shanghai,
			China,
			Email: maoxj@fudan.edu.cn} ~
		Zhuoyi Yang
		\footnote{New York University,
			New York,
			USA,
			Email: zyang@stern.nyu.edu}}
	
	\date{}
	\maketitle
\begin{abstract}
	
	This paper studies distributed estimation and support recovery for high-dimensional linear regression model with heavy-tailed noise. To deal with heavy-tailed noise whose variance can be infinite, we adopt the quantile regression loss function instead of the commonly used squared loss. However, the non-smooth quantile loss  poses new challenges to high-dimensional distributed estimation in both computation and theoretical development. To address the challenge, we transform the response variable and establish a new connection between quantile regression and ordinary linear regression. Then, we provide a distributed estimator that is both computationally and communicationally efficient, where  only the gradient information is communicated at each iteration. Theoretically, we show that, after a constant number of iterations, the proposed estimator achieves a near-oracle convergence  rate without any restriction on the number of machines. Moreover, we establish the theoretical guarantee for the support recovery. The simulation analysis is provided to demonstrate the effectiveness of our method. 

\end{abstract}
\noindent
\textbf{keywords:}
Distributed estimation; high-dimensional linear model; quantile loss; robust estimator; support recovery

\section{Introduction}

{
	The development of internet technology has led to the generation of modern data that exhibits several challenges in statistical estimation:
	\begin{enumerate}
		\item The first challenge comes from the scalability of the data. In particular, modern large-scale data usually cannot be fit into memory or are collected in a distributed environment.  For example, a personal computer usually has a limited memory size in GBs; while the data stored on a hard disk could have a size in TBs. In addition, sensor network data are naturally collected by many sensors. For these types of large-scale data, traditional methods, which load all the data into memory and run a certain optimization procedure (e.g., Lasso),  are no longer applicable due to both storage and computation issues. 
		
		\item The second challenge comes from the dimensionality of data. High-dimensional data analysis has been an important research area in statistics over the past decade. A sparse model is commonly adopted in high-dimensional literature and support recovery is an important task for high-dimensional analysis (see, e.g., \cite{zhao2006model,wainwright2009sharp,buhlmann2011statistics,tibshirani2015statistical}). There are some recent work on statistical estimation for high-dimensional distributed data (see, e.g., \cite{zhao2014general}, \cite{lee2017communication}, \cite{battey2018distributed}). However, these work usually adopt a de-biased approach, which leads to a dense estimated coefficient vector. Moreover, the \emph{support recovery} problem in a distributed setting still largely remains open. 
		
		\item The third challenge comes from heavy-tailed noise, which is prevalent in practice (see, e.g., \cite{hsu2016loss,fan2017estimation,chen2018robust,sun2018adaptive,zhou2018new}). When the finite variance assumption for the noise does not exist, most existing theories based on least squares or Huber loss in robust statistics will no longer be applicable. 
	\end{enumerate}
}

The main purpose of the paper is to provide a new estimation approach for high-dimensional linear regression in a distributed environment and establish the theoretical results on both estimation and support recovery. More specifically, we consider the following linear model, 
\begin{equation}\label{eq:model}
	Y= \X^{\rm T} \be^* +e,
\end{equation}
where $\X=(1, X_{1}, \ldots, X_{p})^{\rm T}$ is a $(p+1)$-dimensional vector,  $\be^*=(\beta^{*}_{0},\beta_1^*,\ldots, \beta^*_{p})^{\rm T}$ is the true regression coefficient, with $\beta^*_0$ being the intercept, and $e$ is the noise. 
We only assume that $e$ is independent of the covariate vector $(X_{1}, \ldots, X_{p})^{\rm T}$ and the density function of  $e$ exists.  { It is worthwhile noting that the independence assumption has been adopted in estimating robust linear models when using a quantile loss function (see, e.g., \cite{zou2008composite,fan2014adaptive}).  In Remark \ref{rmk:independent}, we will briefly comment on how to extend our method to the case when the noise is not independent with covariates.}  Furthermore, we allow the dimension $p$ to be much larger than the sample size $n$ (e.g., $p=o(n^\nu)$ for some $\nu>0$). We assume that $\be^*$ is a sparse vector with $s$ non-zero elements. 


In this paper,  we allow a very heavy-tailed noise $e$, whose variance can be infinite (e.g., Cauchy distribution). For such a heavy-tailed noise, the squared-loss based Lasso approach is no longer applicable. 
To address this challenge, we can assume without loss of generality that $\pr(e\leq 0)=\tau$ for a specified  quantile level $\tau\in (0,1)$ (otherwise, we can shift the first component to be $\beta^*_0-q_{\tau}$ so that this assumption holds, where $q_{\tau}$ is the $\tau$-th quantile of $e$). Then, it is easy to see that
\begin{equation*}
	\be^{*}=\argmin_{\be\in \R^{p+1}}\mathbb{E}\rho_{\tau}(Y-\X^{\rm T}\be),
\end{equation*}
where $\rho_\tau(x)=x(\tau-\ind{x\leq 0})$ (see, e.g., \cite{koenker2005quantile}) is known as the quantile regression (QR) loss function.
Given $n$ \emph{i.i.d.} samples  $(\X_{i},Y_{i})$ for  $1\leq i\leq n$, the high-dimensional QR estimator takes the following form,
\begin{eqnarray}\label{eq:QR}
	\hbe=\argmin\limits_{\be\in\R^{p+1}}\frac{1}{n}\sum_{i=1}^{n}\rho_{\tau}(Y_{i}-\X^{\rm T}_{i}\be)+\lambda_{n}|\be|_{1},
\end{eqnarray}
where $|\be|_{1}$ the $\ell_1$-regularization of $\be$, and $\lambda_n$ is the regularization parameter. 


{ It is worthwhile noting that in robust statistical literature, the MOM (median of means) has been adopted to corrupted data in high-dimensional settings \citep{hsu2014heavy,lugosi2016risk,lecue2017robust,lugosi2017regularization,lecue2018learning}. However, the MOM is a multi-stage method that requires data splitting. Moreover, when true regression coefficients are sparse, support recovery guarantee is not available in existing MOM literature. Moreover, the quantile loss has been a useful approach to deal with heavy-tailed noise, see, e.g., \cite{fan2014adaptive} for single quantile level and \cite{zou2008composite} for multiple quantile levels. However, the existing literature does not address the challenging issue on efficient distributed implementation, which is the main focus of this paper.
}

Although the adoption of QR loss provides robustness to heavy-tailed noises, it also poses new challenges due to limited computation power and memory to store data especially when the sample size and dimension are both large. Therefore, distributed estimation procedure becomes increasingly important. The main purpose of the paper is to develop a new estimation approach for high-dimensional QR and establish the theoretical results on both \emph{estimation} and \emph{support recovery}. In fact, as  we will survey in the next paragraph, the support recovery problem in a high-dimensional distributed setting still largely remains as an open problem.

In a distributed setting, let us assume $n$ samples are stored in $L$ local machines.  In particular, we split the data index set $\{1,2,\ldots,n\}$ into $\mathcal{H}_{1},\ldots,\mathcal{H}_{L}$, where $\mathcal{H}_k$ denotes the set of indices on the $k$-th machine. For the ease of illustration, we assume that the data are evenly distributed ($n/L$ is an integer) and each local machine has the sample size $|\mathcal{H}_k|=m=n/L$  (see Remark \ref{rmk:batchsize} at the end of Section \ref{sec:theory} for the discussion on general data partitions). On each machine, one can construct a local estimator $\hbe_{k}$ by solving
\begin{eqnarray}\label{eq:local}
	\hbe_{k}=\argmin\limits_{\be\in\R^{p+1}}\frac{1}{m}\sum_{i\in\mathcal{H}_{k}}\rho_\tau(Y_{i}-\X^{\rm T}_{i}\be)+\lambda_{m}|\be|_{1}.
\end{eqnarray}
Then the final estimator of $\be^{*}$  can be naturally taken as the averaging estimator $\hbe_{avg}=\frac{1}{L}\sum_{k=1}^{L}\hbe_{k}$. This method is usually known as averaging divide-and-conquer approach (see, e.g., \cite{li2013statistical,zhao2016partially,fan2017distributed,shi2018massive,banerjee2019divide}). Although this method enjoys low communication cost (i.e., one-shot communication), the obtained estimator is usually no longer sparse. Instead of constructing the local estimator in its original form as in \eqref{eq:local}, there are a number of works that construct a de-biased estimator as the local estimator, and then take the average (see, e.g., \cite{zhao2014general,lee2017communication,battey2018distributed}). The de-biased estimator has been popular in high-dimensional statistics (see, e.g., \cite{belloni2013least,van2014asymptotically,zhang2014confidence,javanmard2014confidence} and references therein). \cite{zhao2014general} studied the averaging divide-and-conquer approach for high-dimensional QR based on de-biased estimator. There are several issues of the averaging de-biased estimator for high-dimensional distributed estimation. First, due to de-biasing, the local estimator on each machine is no longer sparse and thus the final averaging estimator cannot be used for support recovery. Second, the de-biased approach needs to estimate a $p\times p$ precision matrix $\S^{-1}$, which requires each machine  to solve $p$ optimization problems (see, e.g., Eq. (3.17) in \cite{zhao2014general}), while each optimization problem involves computing a variant of the CLIME estimator \citep{cai2011constrained}.  In other words, instead of solving one $p$-dimensional optimization as in \eqref{eq:local}, the de-biased estimator requires to solve $(p+1)$  optimization problems. This would be computationally very expensive especially when $p$ is large. Finally, the theoretical result of the averaging estimator requires that the number of machines $L$ is not too large.  For example, in high-dimensional QR, the theoretical development in \cite{zhao2014general} requires $L=o(n^{1/3}/(s \log^{5/3}(\max(p,n))))$, where $s$ is the number of non-zero elements in $\be^*$.
It would be an interesting theoretical question on how to remove such a constraint on $L$. In \cite{wang2017efficient}, \cite{jordan2018communication} and \cite{fan2019communication}, they develop iterative methods with multiple rounds of aggregations (instead of one-shot averaging), which relax the condition on the number of machines. However, their methods and theory require the loss function to be second-order differentiable and thus cannot be applied to the \emph{non-smooth} QR loss. 
We also note that \cite{chen2019} studied distributed QR problem in a low dimensional setting, where $\be^*$ is dense and $p$ grows much more slowly than $n$.

In this paper, we propose a new distributed estimator for estimating high-dimensional linear model with heavy-tailed noise. We first show that the estimation of regression coefficient $\be^*$ can be resorted to a penalized least squares optimization problem with a pseudo-response $\tY_{i}$ instead of $Y_{i}$. This leads to a pooled estimator, which essentially solves a Lasso problem with the squared loss based on $\tY_i$, without requiring any moment condition on the noise term. This pooled estimator is computationally much more efficient than solving high-dimensional QR (\ref{eq:QR})  in a single machine setting. 

Moreover,  our result establishes an interesting connection between the QR estimation and the ordinary linear regression. This connection translates a non-smooth objective function to a smooth one, which greatly facilitates computation in a distributed setting. Given the transformed  penalized least squares formulation, we further provide a communication efficient distributed algorithm, which runs iteratively and only communicates $(p+1)$-dimensional gradient information at each iteration (instead of the $(p+1) \times (p+1)$ matrix information).  Our distributed algorithm  is essentially an approximate Newton method (see, e.g., \cite{shamir2014communication}), which uses gradient information to approximate Hessian information and thus allows efficient communication.  In this paper, we provide a more intuitive derivation of the method simply based on the standard Lasso theory. 

Then we establish the theoretical properties of the proposed distributed  estimator.  We first establish the convergence rate in $\ell_2$-norm for one iteration (Theorem \ref{thm:betainf}).  Based on this result, we further characterize the convergence rate for multiple iterations. We show that, after a constant number of iterations, our method achieves a near-oracle rate of $\sqrt{s\log(\max(p,n))/n}$ (Theorem \ref{thm:betainft}).  This rate is identical to the rate of $\ell_1$-regularized QR in a single machine setting \citep{belloni2011l1}, and almost matches the oracle rate  $\sqrt{s/n}$ (upto a logarithmic factor) where the true support is known. Furthermore, we provide the support recovery result of the distributed estimator. We first show that the estimated support is a subset of the true support with high probability (Theorem \ref{thm:support} and \ref{thm:supportt}). Then we characterize the ``beta-min'' condition for the exact support recovery, and we show that the ``beta-min'' condition becomes weaker as the number of iterations increases (Theorem \ref{thm:supportt}). Again, after a constant number of iterations, the lower bound in our ``beta-min'' condition matches the ideal case with all the samples on a single machine.  To the best of our knowledge, this is the first support recovery result for high-dimensional robust distributed estimation.

\subsection{Paper Organization and Notations} 
The rest of our paper is organized as follows. In Section \ref{sec:method} we define the estimator and provide our algorithm. In Section \ref{sec:theory} we provide the theoretical guarantee for the convergence rate and support recovery for our estimator. Numerical experiments based on simulation are provided in Section \ref{sec:sim} to illustrate the performance of the estimator. Section \ref{sec:conclusion} gives some concluding remarks and future directions. The proofs of main theoretical results is relegated to the Appendix \ref{sec:proofsupp}.

For a vector $\bv=(v_{1},\dots,v_{n})^{\rm T}$, define $\abs{\bv}_{1}=\sum_{i=1}^{n}\abs{v_{i}}$ and $\abs{\bv}_{2}=\sqrt{\sum_{i=1}^{n}v_{i}^{2}}$. For a matrix $\bA=(a_{ij})\in\mbR^{p\times q}$, define $\abs{\bA}_{\infty}=\max_{1\le i\le p,1\le j \le q}\abs{a_{ij}}$, $\norm{\bA}_{L_{1}}=\max_{1\le j\le q}\sum_{i=1}^{p}\abs{a_{ij}}$, $\norm{\bA}_{\mathrm{op}}=\max_{\abs{v}_2=1} \abs{\bA v}_2$, and $\norm{\bA}_{\infty}=\max_{1\le i\le p}\sum_{j=1}^{q}\abs{a_{ij}}$. For two sequences $a_n$ and $b_n$ we say $a_n \asymp b_n$ if and only if both $a_n = O(b_n) $ and $b_n = O(a_n) $ hold. For a matrix $\bA$, define $\Lambda_{\text{max}}(\bA)$ and $\Lambda_{\text{min}}(\bA)$ to be the largest and smallest eigenvalues of $\bA$ respectively. For a matrix $\bA\in \R^{m\times n}$ and two subsets of indices $S=\{s_1,\ldots,s_r\}\subseteq\{1,\ldots,m\}$ and $T = \{t_1,\ldots,t_q\}\subseteq \{1,\ldots,n\}$, we use $\A_{S\times T}$ to denote the $r$ by $q$ submatrix given by $(a_{s_it_j})$.  We use $C,c,c_0,c_1,\ldots$ to denote constants whose value may change from place to place, which do not depend on $n$, $p$, $s$ and $m$. 

\section{Methodology}\label{sec:method}
In this section, we introduce the proposed method. We start with a robust estimator with Lasso (REL), which establishes the connection between quantile regression (QR) and ordinary linear regression in a single machine setting. This proposed estimator  will motivate the construction of our distributed estimator.

\subsection{Robust Estimator with Lasso (REL)}
Our method is inspired by the Newton-Raphson method. Consider the following stochastic optimization problem,
\begin{equation}\label{eq:sto_opt}
	\be^{*}=\argmin_{\be\in \R^{p+1}}\ep [G(\be;\X,Y)],
\end{equation}
where $G(\be;\X,Y)$ is the loss function. In $G(\be;\X,Y)$, $\X$ and $Y$ are random covariates and response and $\be$ is the coefficient vector of interest. To solve this stochastic optimization problem, the population version of the Newton-Raphson iteration takes the following form
\begin{eqnarray}\label{eq:onestep}
	\tilde{\be}_{1}=\be_{0}-\H(\be_{0})^{-1}\ep [g(\be_{0};\X,Y)],
\end{eqnarray}
where $\be_0$ is an initial solution, $g(\be;\X,Y)$ is the subgradient of the loss function $G(\be;\X,Y)$ with respect to $\be$, and $\H(\be):=\partial \ep [g(\be;\X,Y)]/\partial \be $ denotes the population Hessian matrix of $\mathbb{E}G(\be;\X,Y)$. In particular, let us consider the case where $G(\be;\X,Y)$ is the QR loss, i.e.,
\begin{equation}\label{eq:G}
	G(\be;\X,Y) = \rho_\tau (Y-\X^{\rm T}\be).
\end{equation}
Given $G(\be;\X,Y)$ in \eqref{eq:G}, the subgradient and Hessian matrix take the form of
$g(\be;\X,Y)=\X(\ind{Y-\X^{\rm T}\be\leq 0}-\tau)$ and $\H(\be)=\ep (\X\X^{\rm T}f(\X^{\rm T}(\be-\be^{*})))$, respectively. Here, $f(x)$
is the density function of the noise $e$. When the initial estimator $\be_{0}$ is close to the true parameter $\be^{*}$,  $\H(\be_0)$ will be close to $\H(\be^*) = \S f(0)$, where $\S = \ep \X\X^{\rm T}$ is the population covariance matrix of the covariates $\X$. Using $\H(\be^*)$ in \eqref{eq:onestep} motivates the following iteration,
\begin{align}\label{eq:newton}
	\be_{1}=\be_{0}-\H(\be^*)^{-1}\ep [g(\be_{0};\X,Y)]= \be_{0}-\S^{-1}f^{-1}(0)\ep [g(\be_{0};\X,Y)].
\end{align}
Further, under some regularity conditions, we have the following Taylor expansion of $\ep[g(\be_0;\X,Y)]$ at $\be^*$,
\begin{align*}
	\ep[g(\be_0;\X,Y)] =& \H(\be^*)(\be_0-\be^{*})+O(|\be_0-\be^{*}|_2^2)\\
	=&\S f(0)(\be_0-\be^{*})+O(|\be_0-\be^{*}|_2^2).
\end{align*}
Combine it with \eqref{eq:newton}, and it is easy to see that
\begin{align*}
	|\be_{1}-\be^{*}|_{2} =& |\be_0-\S^{-1}f^{-1}(0)\left(\S f(0)(\be_0-\be^{*})+O(|\be_0-\be^{*}|_2^2)\right)-\be^{*}|_2
	\\
	=&O(|\be_{0}-\be^{*}|^{2}_{2}).
\end{align*}
In summary, if we have a consistent estimator $\be_{0}$, we can refine it by the Newton-Raphson iteration in \eqref{eq:newton}.

Next, we show how to translate the Newton-Raphson iteration into a least squares optimization problem. First we rewrite the equation \eqref{eq:newton} to be
\begin{eqnarray*}
	\be_{1}&=&\S^{-1}\Big{(}\S\be_{0}-f^{-1}(0)\ep[g(\be_{0};\X,Y)]\Big{)}\cr
	&=&\S^{-1}\ep \Big{[}\X\Big{\{}\X^{\rm T}\be_{0}-f^{-1}(0)(\ind{Y\leq\X^{\rm T}\be_{0}}-\tau)\Big{\}}\Big{]}.
\end{eqnarray*}
Let us  define a new response variable $\tY$ as 
\begin{eqnarray*}
	\tY=\X^{\rm T}\be_{0}-f^{-1}(0)(\ind{Y\leq\X^{\rm T}\be_{0}}-\tau).
\end{eqnarray*}
Then $\be_{1} = \S^{-1}\ep(\X\tY)$ is the best linear regression coefficient of $\tY$ on $\X$, i.e., $\be_{1}=\argmin_{\be\in \R^{p+1}}\ep (\tY-\X^{\rm T}\be)^{2}$.
To further encourage the sparsity of the estimator, it is natural to consider the following $\ell_1$-regularized problem,
\begin{eqnarray}\label{eq:l1}
	\be_{1,\lambda}=\argmin_{\be\in \R^{p+1}} \frac{1}{2}\ep (\tY-\X^{\rm T}\be)^{2}+\lambda|\be|_{1},
\end{eqnarray}
where $\be_{1,\lambda}$ is sparse and can achieve a better convergence rate than $\be_{0}$. So far, we have shown that if we have a consistent estimator $\be_{0}$ of $\be^{*}$, then the estimation of the high-dimensional sparse $\be^{*}$ can be implemented by solving a penalized least squares optimization in \eqref{eq:l1} instead of the penalized QR optimization.
It is well known that the latter optimization problem is computationally expensive when $n$ is large since the QR loss is non-smooth.  More importantly, the transformation from QR loss to least squares will greatly facilitate the development of the distributed estimator. In particular, our distributed estimator is derived from the Lasso theory, which is based on the squared loss (see Section \ref{sec:dist}).

Now, we are ready to define the empirical version of $\be_{1,\lambda}$ in a single machine setting. Let $\hbe_{0}$ be an initial estimator of $\be^{*}$ and  $\widehat{f}(0)$ be an estimator of the density $f(0)$. We use $\hbe_{0}$ to denote the empirical version of the initial estimator, which is distinguished from the population version $\be_0$.  Given $n$ \emph{i.i.d.} samples $(\X_i,Y_i)$ from \eqref{eq:model}, for each $1\le i\le n$, we construct
\begin{eqnarray*}
	\tY_{i}=\X^{\rm T}_{i}\hbe_{0}-\widehat{f}^{-1}(0)(\ind{Y_{i}\leq\X^{\rm T}_{i}\hbe_{0}}-\tau).
\end{eqnarray*}
It is natural to estimate $\be^{*}$ by the empirical version of \eqref{eq:l1}:
\begin{eqnarray}\label{eq:pool}
	\hbe_{pool}=\argmin\limits_{\be\in\RR^{p+1}}\Big{\{} \frac{1}{2n}\sum_{i=1}^{n}(\widetilde{Y}_{i}-\X^{\rm T}_{i}\be)^{2}+\lambda_{n}|\be|_{1}\Big{\}}.
\end{eqnarray}
We note that in a single machine setting, computing this pooled estimator essentially solves a Lasso problem, which is computationally much more efficient than solving an $\ell_1$-regularized QR problem. 

Finally,  we choose $\widehat{f}(0)$ to be a kernel density estimator of $f(0)$:
\begin{eqnarray*}
	\widehat{f}(0)=\frac{1}{nh}\sum_{i=1}^{n}K\Big{(}\frac{Y_{i}-\X^{\rm T}_{i}\widehat{\be}_{0}}{h}\Big{)},
\end{eqnarray*}
where $K(x)$ is a kernel function which satisfies the condition (C3) (see Section 3) and $h\to 0$ is the bandwidth. The selection of bandwidth will be discussed in our theoretical results (see Section \ref{sec:theory}).

In the next section, we will introduce a distributed robust estimator with Lasso which can estimate $\be^{*}$ with a near-oracle convergence rate.

\subsection{Distributed Robust Estimator with Lasso}
\label{sec:dist}

Given our new proposed estimator $\hbe_{pool}$, we can use the approximate Newton method to solve the distributed estimation problem. To illustrate this technique from the Lasso theory, we first consider a general convex quadratic optimization as follows,
\begin{eqnarray}\label{eq:bhat}
	\widehat{\be}=\argmin\limits_{\be\in\RR^{p+1}}\frac{1}{2}\be^{\rm T}\A\be-\be^{\rm T}\b+\lambda_{n}|\be|_{1},
\end{eqnarray}
where $\A$ is a non-negative definite matrix and $\b$ is a vector in $\RR^{p+1}$.
From standard Lasso theory (see \cite{buhlmann2011statistics}), we have the following proposition.
\begin{proposition}\label{prop:0}
	Assume the following conditions hold
	\begin{eqnarray}\label{cd1}
		|\A\be^{*}-\b|_{\infty}\leq \lambda_{n}/2,
	\end{eqnarray}
	\begin{eqnarray}\label{cd2}
		\min_{\delta: |\delta|_{1}\leq c_{1}\sqrt{s}|\delta|_{2}}\frac{\delta^{\mathrm{T}}\A\delta}{|\delta|^{2}_{2}}\geq c_{2},\quad c_{1},c_{2}>0.
	\end{eqnarray}
	where $s$ is the sparsity of $\be^{*}$, i.e., $s=\sum_{j=0}^{p}\ind{\beta^{*}_{j}\neq 0}$. Then we have
	\begin{eqnarray}\label{prop4}
		|\hbe-\be^{*}|_{2}\leq c\sqrt{s}\lambda_{n}, 
	\end{eqnarray}
	for some constant $c>0$.
\end{proposition}

Note that the condition \eqref{cd2} is known as the compatibility condition, which is used to provide the $\ell_2$-consistency of the Lasso estimator. For the purpose of completeness, we include a proof of Proposition \ref{prop:0} in the Appendix \ref{sec:proofsupp}. As one can see from \eqref{cd1}, if we can choose a matrix $\A$ and a vector $\b$ such that $\lambda_{n}$ is as small as possible, we can obtain a fast convergence rate of $\hbe$. 

Now let us discuss how to use Proposition \ref{prop:0} to develop our distributed estimator. Suppose that $n$ samples are stored in $L=n/m$ machines and each local machine has $m$ samples. We first split  the data index set $\{1,2,\ldots,n\}$ into $\mathcal{H}_{1},\ldots,\mathcal{H}_{L}$ with  $|\mathcal{H}_{k}|=m$ and the $k$-th machine stores samples $\{(\X_{i},Y_i): \; i\in \mathcal{H}_{k}\}$. 
Let us define
\begin{eqnarray}\label{eq:local_hessian}
	\hS_{k}=\frac{1}{m}\sum_{i\in \mathcal{H}_{k}}\X_{i}\X^{\rm T}_{i},\quad\hS = \frac{1}{n}\sum_{i=1}^n \X_i\X_i^{\rm T} = \frac{1}{L} \sum_{k=1}^L \hS_k,
\end{eqnarray}
as the sample covariance matrix on the $k$-th machine and the sample covariance matrix of the entire dataset, respectively. It is worthwhile noting that our algorithm does not need to explicitly compute and communicate $\hS_{k}$ (for $k \neq 1$) (see Algorithm \ref{alg:1} for more details).

In Proposition \ref{prop:0}, we first choose $\A = \hS_1$ to be the sample covariance matrix computed on the first machine. Our goal is to construct a vector $\b$ such that $|\A\be^* - \b|_\infty $ can be as small as possible. Note that
\begin{align}\label{eq:Ab}
	\nonumber\A\be^*-\b =& \hS_1\be^*-\b\\
	=& \hS \be^* +(\hS_1-\hS)\be^* - \b.
\end{align}
It can be proved that $\hS\be^*$ is close to $\z_n := \frac{1}{n}\sum_{i=1}^n \X_i\tY_i$ (see Proposition \ref{prop:Bnbeta0} in the Appendix \ref{sec:proofsupp}). We note that $\z_n$ can be computed effectively in a distributed setting since 
$$\z_n = \frac{1}{L} \sum_{k=1}^L \z_{nk},\quad\z_{nk}=\frac{1}{m}\sum_{i\in \mathcal{H}_{k}}\X_{i}\tY_{i},$$
where 
$\z_{nk}$ can be computed on the $k$-th local machine. Therefore we can rewrite $\eqref{eq:Ab}$ as
\begin{equation*}
	\begin{aligned}
		|\A\be^*-\b|_{\infty} =& |\hS\be^*-\z_n+\z_n+(\hS_1-\hS)\be^*-\b|_{\infty}\\
		\leq & |\hS\be^*-\z_n|_{\infty}+|\z_n+(\hS_1-\hS)\be^*-\b|_{\infty}.
	\end{aligned}
\end{equation*}
Since $\be^*$ is unknown, in order to make the second term as small as possible, it is natural to set $$\b = \z_n +(\hS_1-\hS)\hbe_0.$$ For $\A = \hS_1$ and $\b = \z_n +(\hS_1-\hS)\hbe_0$, we can prove that (see Eq. \eqref{easytoshow} in the proof of Theorem \ref{thm:betainf} and \ref{thm:betainft})
\begin{eqnarray*}
	|\hS_{1}\be^{*}-\b|_{\infty}\leq \lambda_{n}/2,
\end{eqnarray*}
for some specified $\lambda_{n}$ (see Theorem \ref{thm:betainf}). With $\A$ and $\b$ in place, the equation \eqref{eq:bhat} leads to the following $\ell_1$-regularized quadratic programming,
\begin{align}\label{eq:beta_dist}
	\hbe^{(1)}=\argmin\limits_{\be\in\RR^{p+1}}\frac{1}{2m}\sum_{i\in \mathcal{H}_1}(\X^{\rm T}_{i}\be)^{2}-\be^{\rm T}\Big{\{}\z_{n}+(\hS_{1}-\hS)\hbe_{0}\Big{\}}+\lambda_{n}|\be|_{1}.
\end{align}
Note that when $m=n$, we have $\hbe^{(1)}=\hbe_{pool}$. In other words, when the data is pooled on a single machine, the proposed distributed estimator automatically reduces to $\hbe_{pool}$ in \eqref{eq:pool}. We also note that $\hS \hbe_0$ in the vector $\b$ can be computed effectively in a distributed manner. In particular, each local machine computes and communicates a $(p+1)$-dimensional vector $\hS_k \hbe_{0}=\frac{1}{m}\sum_{i\in \mathcal{H}_{k}}\X_{i}(\X^{\rm T}_{i}\hbe_{0})$ to the first machine. Then the first machine computes $\hS \hbe_0$ by
\[
\hS \hbe_0 =\frac{1}{L} \sum_{k=1}^L \hS_k \hbe_{0}.
\]
Our algorithm only communicates $\z_{nk} = \frac{1}{m}\sum_{i\in \mathcal{H}_k} \X_i\tY_i$ and $\hS_k\hbe_0$ to the first machine at each iteration. Therefore, the per-iteration communication complexity is only $O(p)$ and there is no need to communicate the $(p+1)\times (p+1)$ sample covariance matrix $\hS_k$.

Given \eqref{eq:beta_dist} as the estimator from the first iteration, it is easy to construct an iterative estimator. In particular, let $\widehat{\be}^{(t-1)}$ be the distributed REL in the $(t-1)$-th iteration. Define
\[
\widehat{f}^{(t)}\left(0\right)=\frac{1}{nh_{t}}\sum_{i=1}^{n}K\left(\frac{Y_{i}-\bX_{i}^{\rm T}\bhbeta^{(t-1)}}{h_{t}}\right),
\]
as the density estimator in the $t$-th iteration where $h_{t}\to 0$ is the bandwidth for the $t$-th iteration. The bandwidth $h_t$ shrinks as $t$ grows, whose rate will be specified in Theorem \ref{thm:betainft}.
Let us define 
\begin{equation}\label{eq:ytilde}
	\tY_{i}^{(t)}=\bX_{i}^{\rm T}\bhbeta^{(t-1)}-(\widehat{f}^{(t)}\left(0\right))^{-1}\left(\Ind{Y_{i}\le \bX_{i}^{\rm T}\bhbeta^{(t-1)}}-\tau\right),
\end{equation}
and
\[
\z_{n}^{(t)}=\frac{1}{n}\sum_{i=1}^{n}\bX_{i}\tY_{i}^{(t)}. 
\]
As in \eqref{eq:beta_dist}, our distributed estimator $\bhbeta^{(t)}$ is the solution of the following $\ell_1$-regularized quadratic programming problem:
\begin{align}\label{eq:betat}
	\hbe^{(t)}=\argmin_{\be\in\mbR^{p+1}}\frac{1}{2m}\sum_{i\in \mathcal{H}_1}(\X^{\rm T}_{i}\be)^{2}-\be^{\rm T}\left\{\z_{n}^{(t)}+\left(\bhSigma_{1}-\bhSigma\right)\hbe^{(t-1)}\right\}+\lambda_{n,t}\Abs{\be}_{1}.
\end{align}
It is worthwhile noting that the convex optimization problem \eqref{eq:betat} has been extensively studied in the optimization literature and several efficient optimization methods have been developed, e.g., FISTA \citep{beck2009fast}, active set method \citep{solntsev2015algorithm}, and PSSgb (Projected Scaled Subgradient,  Gafni-Bertsekas variant, \citep{schmidt2010graphical}). In our experiments, we adopt the PSSgb optimization method for solving \eqref{eq:betat}. We present the entire distributed estimation procedure in Algorithm \ref{alg:1}.

\begin{algorithm}[!t]
	\caption{{\small Distributed high-dimensional QR estimator}}
	\label{alg:1}
	\hspace*{\algorithmicindent} \hspace{-0.72cm}  {\textbf{Input:} Data on local machines $\{\X_i,Y_i:\;i\in \mathcal{H}_k\}$ for $k=1,\ldots, L$, the number of iterations $t$, quantile level $\tau$, kernel function $K$, a sequence of bandwidths $h_g$ for $g=1,\ldots, t$ and the regularization parameters $\lambda_0$, $\lambda_{n,g}$ for $g=1,\ldots,t$.}
	
	\begin{algorithmic}[1]
		\State Compute the initial estimator $\hbe^{(0)} = \hbe_0$ based on $\{\X_i,Y_i:\;i\in \mathcal{H}_1\}$:
		\begin{eqnarray}\label{ag0}
			\hbe_{0} = \argmin\limits_{\be\in\R^{p+1}}\frac{1}{m}\sum_{i\in\mathcal{H}_{1}}\rho_\tau(Y_{i}-\X^{\rm T}_{i}\be)+\lambda_{0}|\be|_{1}.
		\end{eqnarray}
		\For{$g=1,2 \ldots, t$}
		\State Transmit $\hbe^{(g-1)}$ to all local machines.
		\For{$k=1,\dots, L$}
		\State The $k$-th machine computes $	\widehat{f}^{(g,k)}\left(0\right):=\frac{1}{m}\sum_{i\in \mathcal{H}_k}K\left(\frac{Y_{i}-\bX_{i}^{\rm T}\bhbeta^{(g-1)}}{h_{g}}\right)$ and sends it back to the first machine.
		\EndFor
		\State The first machine computes $\widehat{f}^{(g)}\left(0\right)$ based on 
		\[
		\widehat{f}^{(g)}\left(0\right)=\frac{1}{L}\sum_{k=1}^L\widehat{f}^{(g,k)}\left(0\right).
		\]
		\State Transmit $\widehat{f}^{(g)}\left(0\right)$ to all local machines.
		\For{$k=1,\dots, L$}
		\State The $k$-th machine computes $\hS_k\hbe^{(g-1)}$ and $\z_{nk}=\frac{1}{m}\sum_{i\in \mathcal{H}_k}\X_i\tY_i^{(g)}$ based on \eqref{eq:ytilde} and sends them back to the first machine.
		\EndFor
		\State Compute the estimator $\widehat{\be}^{(g)}$ on the first machine based on \eqref{eq:betat}.
		\EndFor
		
	\end{algorithmic}
	\textbf{Output:}  The final estimator $\widehat{\be}^{(t)}$.
\end{algorithm}

For the choice of the initial estimator $\hbe_{0}$, we propose to solve the high-dimensional QR problem using the data on the first machine, i.e.,
\begin{equation}\label{eq:init}
	\begin{aligned}
		\hbe_{0} = \argmin\limits_{\be\in\R^{p+1}}\frac{1}{m}\sum_{i\in\mathcal{H}_{1}}\rho_\tau(Y_{i}-\X^{\rm T}_{i}\be)+\lambda_{0}|\be|_{1}.
	\end{aligned}
\end{equation}
Note that although this paper uses the \eqref{eq:init} as the initial estimator, one can adopt any estimator as $\hbe_0$ as long as it satisfies the condition (C6) (see Section \ref{sec:theory}).

We assume the quantile level $\tau$ is pre-specified in Algorithm \ref{alg:1}. Our paper mainly focuses on the algorithm for distributed estimation under a general $\tau$ and develop the related theoretical results.  Different choices of $\tau$ correspond to different loss functions we want to use and different parameters we are interested in.  The choice of  $\tau$ to fit the model is a separate topic which clearly depends on the practical problem  and the parameters we are interested in. For example, without the covariate $\bm{X}$ (for briefness),  $\beta^{*}_{0}$ is the $\tau$-quantile of $Y$ and the choice of $\tau$ depends on what  quantile of $Y$ we are interested in. In extreme climate studies, people would like to choose $\tau$ as some large values ($0.9$ and $0.99$) or small values ($0.1$ and $0.01$) to evaluate the extreme climate performance. In economic domain, to learn the problem associated with median salary, we can simply set $\tau=0.5$.


\section{Theoretical Results}\label{sec:theory}
In this section we provide the theoretical results for our distributed method. We define
\begin{eqnarray*}
	S=\{0\leq i\leq p: \beta^{*}_{i}\neq 0\},
\end{eqnarray*}
as the support of $\be^*$ and $s=|S|$.
We assume the following regular conditions.

\vspace{3mm}

(C1) The density function of the noise $f(\cdot)$ is bounded and Lipschitz continuous (i.e., $\abs{f(x)-f(y)}\le C_L\abs{x-y}$ for any $x,y\in\mbR$ and some constant $C_L>0$). Moreover, we assume $f(0)>c>0$ for some constant $c$.

(C2) Suppose that $\bSigma=\ep \X\X^{\mathrm{T}}$ satisfies
\begin{equation}\label{eqn:irres}
	\Norm{\bSigma_{S^{c}\times S}\bSigma_{S\times S}^{-1}}_{\infty}\le 1-\alpha,
\end{equation}
for some $0<\alpha<1$. Also assume that $c_{0}^{-1}\le\Lambda_{\text{min}}(\bSigma)\le\Lambda_{\text{max}}(\bSigma)\le c_{0}$ for some constant $c_{0}>0$. 

(C3) Assume that the kernel function $K(\cdot)$ is integrable with $\int_{-\infty}^\infty K(u)\mathrm{d}u = 1$. Moreover, assume that $K(\cdot)$  satisfies $K(u)=0$ if $|u|\ge 1$. Further, assume $K(\cdot)$ is differentiable and its derivative $K'(\cdot)$ is bounded. 

(C4) We assume that the covariate $\bX$ satisfies the sub-Gaussian condition for some $t>0$ and $C>0$, $$\sup_{\abs{\btheta}_{2}=1}\mbE\exp(t(\btheta^{\rm T}\bX)^2)\le C.$$

(C5) The dimension $p$ satisfies $p=O(n^{\nu})$ for some $\nu>0$. The local sample size $m$ on each machine satisfies $m\geq n^{c}$ for some $0<c<1$, and the sparsity level $s$ satisfies $s=O(m^{r})$ for some $0<r<1/3$.

(C6) The initial estimator $\bhbeta_{0}$ satisfies $\abs{\bhbeta_{0}-\bbeta^{*}}_{2}=O_{\pr}(\sqrt{s(\log n)/m})$. Furthermore, assume that $\pr(\text{supp}(\bhbeta_{0})\subseteq S)\rightarrow 1$.\vspace{2mm}

Condition (C1) is a regular condition on the smoothness of the density function $f(\cdot)$. { Condition (C2) is the standard irrepresentable condition, which is commonly adopted to establish support recovery in high-dimensional statistics literature (see, e.g., \cite{zhao2006model,wainwright2009sharp,buhlmann2011statistics,tibshirani2015statistical})}. 
Condition (C3) is a standard condition on the kernel function $K(\cdot) $ (see an example of $K(\cdot)$ in Section \ref{sec:sim}).  
Condition (C4) is a regular condition on the distribution of $\X$ while Condition (C5) is on dimension $p$, local sample size $m$ and sparsity level $s$.  
The conditions $m \geq n^c$ for some $0 < c <1$ and $s=O(m^r)$ make sure that our algorithm achieves the near-oracle convergence rate only using a finite number of iterations (see Eq.~\eqref{eq:t} below).  Condition (C6) is a condition on the convergence rate and support recovery of the initial estimator. Note that in Algorithm \ref{alg:1}, the initial estimator $\hbe_0$ is proposed as the solution to the high-dimensional QR problem using data on the first machine. 
It can be shown that $\hbe_0$ in \eqref{eq:init} fulfills condition (C6) under conditions (C1), (C2), (C4), (C5) and some regularity conditions \citep{fan2014adaptive}. In addition, we also show that the condition (C6) is satisfied for the proposed estimator for the $t$-th iteration $\widehat{\beta}^{(t)}$, which serves as the initial estimator for the $(t+1)$-th iteration, in Theorems \ref{thm:betainf}--\ref{thm:supportt}. We also note that by $p=O(n^\nu)$ in (C5), we have that $\log(\max(n,p))=C_1 \log(n)$ for some constant $C_1>0$. Therefore, we will use $\log(n)$ in our convergence rates (instead of $\log(\max(n,p))$) for notational simplicity.

Let $\{a_{n}\}$ be the convergence rate of the initial estimator, i.e., $\abs{\bhbeta_{0}-\bbeta^{*}}_{2}=O_{\pr}(a_{n})$. By condition (C6) we can assume that $a_{n}=\sqrt{s(\log n)/m}$. We first provide the convergence rate for $\hbe^{(1)}$ after one iteration.

\begin{theorem}\label{thm:betainf}
	Let $\abs{\bhbeta_{0}-\bbeta^{*}}_{2}=O_{\pr}(a_{n})$ and choose the bandwidth $h\asymp a_{n}$, take 
	$$\lambda_n=C_{0}\left(\sqrt{\frac{\log n}{n}}+a_{n}\sqrt{\frac{s\log n}{m}}\right),$$ with $C_{0}$ being a sufficiently large constant. Under (C1)-(C6), we have
	\begin{equation}\label{eqn:betainf}
		\Abs{\bhbeta^{(1)}-\bbeta^{*}}_{2}=O_{\pr}\left(\sqrt{\frac{s\log n}{n}}+a_{n}\sqrt{\frac{s^{2}\log n}{m}}\right).
	\end{equation}
\end{theorem}

With the choice of the bandwidth $h$ shrinking at the same rate as $a_{n}$, conclusion \eqref{eqn:betainf} shows that one iteration enables a refinement of the estimator with its rate improved from $a_{n}$ to $\max\{\sqrt{s(\log n)/n},a_{n}\sqrt{s^{2}(\log n)/m}\}$ where $ \sqrt{s^{2}(\log n)/m} = o(1)$ by condition (C5). By recursive applications of Theorem \ref{thm:betainf}, we provide the convergence rate for the multi-iteration estimator $\hbe^{(t)}$. The next theorem shows that an iterative refinement of the initial estimator will improve the estimation accuracy and achieve a near-oracle rate after a constant number of iterations.

In particular, let us define
\begin{eqnarray}\label{eq:a}
	a_{n,g}=\sqrt{\frac{s \log n}{n}}+s^{(2g+1)/2}\left(\frac{\log n}{m}\right)^{(g+1)/2},\quad 0\leq g\leq t.
\end{eqnarray}
From Theorem \ref{thm:betainft} below, we can see that $a_{n,g}$ is the convergence rate of the estimator $\hbe^{(g)}$ after $g$ iterations.

\begin{theorem}\label{thm:betainft}
	Assume that the initial estimator $\bhbeta_{0}$ satisfies $\abs{\bhbeta_{0}-\bbeta^{*}}_{2}=O_{\pr}(\sqrt{s(\log n)/m})$. Let $h_{g}\asymp a_{n,g-1}$ for $1\le g\le t$, and take
	\begin{equation}\label{eq:lambda}
		\begin{aligned}
			\lambda_{n,g}=C_{0}\left(\sqrt{\frac{\log n}{n}}+a_{n,g-1}\sqrt{\frac{s\log n}{m}}\right),
		\end{aligned}
	\end{equation} with $C_{0}$ being a sufficiently large constant. Under (C1)-(C6), we have
	\begin{equation}\label{eq:bt}
		\Abs{\bhbeta^{(t)}-\bbeta^{*}}_{2}=O_{\pr}\left(\sqrt{\frac{s \log n}{n}}+s^{(2t+1)/2}\left(\frac{\log n}{m}\right)^{(t+1)/2}\right).
	\end{equation}	
\end{theorem}

It can be shown that when the iteration number $t$ is sufficiently large, i.e.,
\begin{equation}\label{eq:t}
	t\geq \frac{\log (n/m)}{\log (c_0m/(s^2\log n))},\quad \text{for some }c_0>0,
\end{equation}
the second term in \eqref{eq:bt} is dominated by the first term, and the convergence rate in \eqref{eq:bt} becomes $\abs{\bhbeta^{(t)}-\bbeta^{*}}_{2}=O_{\pr}(\sqrt{s(\log n)/n})$. We note that this rate matches the convergence rate of the $\ell_1$-regularized QR estimator in a single machine setup (see \cite{belloni2011l1}). Moreover, it  nearly matches the oracle convergence rate $\sqrt{s/n}$ (upto a logarithmic factor)  when the support of $\be^*$ is known. We also note that the conditions $m \geq n^c$ and $s = o(m^{1/3})$ in (C5) ensure that the right hand side of \eqref{eq:t} is bounded by a constant, which implies that a constant number of iterations would guarantee a near-oracle rate of $\hbe^{(t)}$.

The following theorems provide results on support recovery of the proposed estimators $\hbe^{(1)}$ and $\hbe^{(t)}$. Recall $S=\{j:\beta^{*}_{j}\neq0\}$ is the support of $\bbeta^{*}$. Let $\hbe^{(1)}=(\widehat{\beta}_{0}^{(1)},\widehat{\beta}_{1}^{(1)},\ldots,\widehat{\beta}_{p}^{(1)})^{\mathrm{T}}$ and
\[
\widehat{S}^{(1)}=\left\{j:\widehat{\beta}_{j}^{(1)}\neq0\right\}.
\]
\begin{theorem}\label{thm:support}
	Assume that the conditions in Theorem \ref{thm:betainf} hold.
	
	(i) We have $\widehat{S}^{(1)}\subseteq S$ with probability tending to one.
	
	(ii) In addition, suppose that for a sufficiently large constant $C>0$,
	\begin{equation}\label{eqn:sigcon}
		\underset{j\in S}{\min}\Abs{\beta^{*}_{j}}\ge C\|\S^{-1}_{S\times S}\|_{\infty}\left(\sqrt{\frac{\log n}{n}}+a_{n}\sqrt{\frac{s\log n}{m}}\right).
	\end{equation}
	Then we have $\widehat{S}^{(1)}= S$ with probability tending to one.
\end{theorem}

Based on Theorem \ref{thm:support}, we can further obtain the support recovery result for $\hbe^{(t)}$, which requires a weaker condition on $\underset{j\in S}{\min}\Abs{\beta^{*}_{j}}$.
Denote $\hbe^{(t)}=(\widehat{\beta}_{0}^{(t)},\widehat{\beta}_{1}^{(t)},\ldots,\widehat{\beta}_{p}^{(t)})^{\mathrm{T}}$ and
\[
\widehat{S}^{(t)}=\left\{j:\widehat{\beta}_{j}^{(t)}\neq0\right\}.
\]
\begin{theorem}\label{thm:supportt} Assume the conditions in Theorem \ref{thm:betainft} hold.
	
	(i) We have $\widehat{S}^{(t)}\subseteq S$ with probability tending to one.
	
	(ii) In addition, suppose that for a sufficiently large constant $C>0$,
	\begin{equation}\label{eqn:sigcont}
		\underset{j\in S}{\min}\Abs{\beta^{*}_{j}}\ge C\|\S^{-1}_{S\times S}\|_{\infty}\left(\sqrt{\frac{\log n}{n}}+s^{t}\left(\frac{\log n}{m}\right)^{(t+1)/2}\right).
	\end{equation}
	Then we have $\widehat{S}^{(t)}= S$ with probability tending to one.
\end{theorem}

Note that the ``beta-min'' condition gets weaker as $t$ increases. When $t$ satisfies \eqref{eq:t}, the condition \eqref{eqn:sigcont} will reduce to $\underset{j\in S}{\min}\Abs{\beta^{*}_{j}}\ge C\|\S^{-1}_{S\times S}\|_{\infty}\sqrt{\frac{\log n}{n}}$, which matches the rate of the lower bound for the ``beta-min'' condition in Lasso in a single machine setting (see \cite{wainwright2009sharp}). 

Furthermore, we state the results in both Theorem \ref{thm:support} and \ref{thm:supportt} by a high-probability statement ``with probability tending to one''. The convergence rate actually can be represented as $1-q_{n}$, where $q_{n}=O(1-\mathbb{P}(\text{supp}(\hat{\bm{\beta}}_{0})\subseteq S))+O(n^{-\gamma})$ is a small quantity goes to $0$ when both $n$ and $p$ go to $\infty$. More specifically, the convergence rate  depends on the convergence rate $\mathbb{P}(\text{supp}(\hat{\bm{\beta}}_{0})\subseteq S)\rightarrow 1$ for the initial estimator $\hat{\bm{\beta}}_0$. Below we further provide two remarks on our method.

\begin{remark}\label{rmk:batchsize}
	It is worthwhile noting that we assume the data is evenly split only for the ease of discussions. In fact, the local sample size $m$ in our theoretical results is the sample size on the first machine  in Algorithm \ref{alg:1} (a.k.a. the central machine in distributed computing).  As long as the sample size $m$ on the first machine is specified, our method does not depend on the partition of the entire dataset.
\end{remark}

{
	\begin{remark}\label{rmk:independent}
		We note that the proposed estimator can be generalized to the case when the noise $e$ and the covariates $\X$ are not independent. More specifically, without the independence assumption, we assume $\pr(e\leq 0 |\X) = \tau$ for some specified $\tau \in (0,1)$. The Hessian matrix becomes $\H(\be^*)=\ep (\X\X^{\rm T}f(0|\X))$. Although $\H(\be^*)$ no longer takes the form of $\S f(0)$ when the noise depends on covariates, it can be approximate by
		\[
		\D_h(\be_0)=
		\ep\left(\X\X^{\rm T}\frac{1}{h}K\left(\frac{Y-\X^{\rm T}\be_0}{h}\right)\right),
		\] 
		for a positive kernel function $K(\cdot)$ (i.e., $K(x) >0$ for all $x$). Let $\hbe_{0}$ be an initial estimator of $\be^{*}$.  Given $n$ \emph{i.i.d.} samples $(\X_i,Y_i)$ from \eqref{eq:model}, for each $1\le i\le n$, we construct the following quantities:
		\[
		\gamma_{i,h} = \sqrt{\frac{1}{h}K\left(\frac{Y_i-\X_i^{\rm T}\hbe_0}{h}\right)},\quad  \tX_{i,h} =  \gamma_{i,h} \X_i,\quad \hD_h = \frac{1}{n} \sum_{i=1}^n \tX_{i,h}\tX_{i,h}^{\rm T},
		\]
		\begin{eqnarray*}
			\tY_{i,h}=\tX_{i,h}^{\rm T}\hbe_{0}-\frac{\ind{Y_i\leq\X_i^{\rm T}\hbe_{0}}-\tau}{\gamma_{i,h}}.
		\end{eqnarray*}
		Then, we can construct the pooled estimator (i.e., the counterpart of \eqref{eq:pool}) by solving the following Lasso problem with both transformed input $\tX_{i,h}$ and response $\tY_{i,h}$:
		\begin{eqnarray}\label{eq:dependent}
			\hbe=\argmin\limits_{\be\in\RR^{p+1}}\Big{\{} \frac{1}{2n}\sum_{i=1}^{n}(\widetilde{Y}_{i,h}-\tX^{\rm T}_{i,h}\be)^{2}+\lambda_{n}|\be|_{1}\Big{\}}.
		\end{eqnarray}
		Using a similar distributed approach described in Section \ref{sec:dist},  the pooled estimator in Eq. \eqref{eq:dependent} can be extended into a distributed estimator. 
		
		Although the extension to the dependent case seems relatively straightforward, the nonparametric estimation of the conditional density $f(0|\X)$ has the issue of ``curse of dimensionality'', especially when $\X$ is high-dimensional. Without any strong assumption on $f(0|\X)$, it requires a huge number of local samples to construct an accurate estimator $\hD_{1,h} = \frac{1}{m} \sum_{i\in \mathcal{H}_1} \tX_{i,h}\tX_{i,h}^{\rm T}$ in the distributed implementation.  We leave more investigation of the dependent noise case to future work.
		
	\end{remark}
}

\section{Simulation Study}\label{sec:sim}
In this section, we report  the simulation studies to illustrate the performance of our distributed REL.
\subsection{Simulation Setup}
We consider the following linear model
\[
Y_i = \X_i^\mathrm{T}\be^* +e_i, \quad i=1,2,\ldots,n,
\]
where $\X_i^\mathrm{T}=(1,X_{i,1},\ldots,X_{i,p})$ is a $(p+1)$-dimensional covariate vector and\linebreak[4] $(X_{i,1},\ldots,X_{i,p})$s are drawn $i.i.d.$ from a multivariate normal distribution $N(0,\S)$. The covariance matrix $\S$ is constructed by $\S_{ij} = 0.5 ^{|i-j|}$ for $1\leq i,j\leq p$. We fix the dimension $p=500$ and choose the loss function to be the QR loss with quantile level $\tau = 0.3$. Note that other choices of $\tau$ lead to similar results in the experiment. We provide additional experimental results for $\tau=0.5$ in the appendix. Let $s$ be the sparsity level and the true coefficient is set to $$\be^* = (\frac{10}{s},\frac{20}{s},\frac{30}{s},\ldots,\frac{10(s-1)}{s},10,0,0\ldots,0).$$ We consider the following three noise distributions:
\begin{enumerate}
	\item Normal: the noise $e_i \sim \mathrm{N}(0,1)$.
	\item Cauchy: the noise $e_i\sim \mathrm{Cauchy}(0,1)$.
	\item Exponential: the noise $e_i\sim \mathrm{exp}(1)$.
\end{enumerate} 
We note that the variance of the Cauchy distribution is infinite.
The initial estimator is computed by directly solving  the $\ell_1$-regularized QR optimization using only the data on the first machine (see Eq. \eqref{ag0}). At each iteration, 
the constant $C_0$ in the regularization parameter $\lambda_{n,g}$ in \eqref{eq:lambda} is chosen by validation. In particular, we choose $C_0$ to minimize the quantile loss on an independently generated validation dataset with the sample size $n$. Moreover, we could also apply cross-validation or an information criterion such as BIC to choose $\lambda_{n}$.

For the choice of the kernel function $K(\cdot)$, we use a biweight kernel function
\[
K(x) = \begin{cases}
	0, & \text{if} \quad x\leq -1,\\
	-\frac{315}{64}x^6+\frac{735}{64}x^4-\frac{525}{64}x^2+\frac{105}{64}, & \text{if} \quad -1\leq x\leq 1,\\
	0, & \text{if} \quad x \geq 1.
\end{cases}
\]
It is easy to verify that $K(\cdot)$ satisfies the condition (C3). We also note that other choices of $K(\cdot)$ provide similar results. 

From Theorem \ref{thm:betainf} and \ref{thm:betainft} in Section \ref{sec:theory}, the bandwidth is set to $h_g =ca_{n,g-1}$ for some constant $c>0$, where $a_{n,g-1}$ is defined in \eqref{eq:a}. In our simulation study, we choose $h_g = \sqrt{\frac{s \log n}{n}}+s^{-1/2}\left(c_0\frac{s^2\log n}{m}\right)^{(g+1)/2}$ (i.e., set the constant $c=1$) for convenience. Note that the constant $c_0$ is used to ensure that $\frac{s^2\log n}{m}<1$, and we set $c_0=0.1$ in the following experiments. In fact, our algorithm is quite robust with respect to the choice of the bandwidth (see the sensitivity analysis in Section \ref{sec:sensitivity}). All the results reported in this section are average of 100 independent runs of simulations.

We compare the performance of the proposed distributed REL (dist REL for short) with other two approaches:
\begin{enumerate}
	\item Averaging divide-and-conquer (Avg-DC) which computes the $\ell_1$-regularized QR (see Eq. \eqref{eq:local}) on each local machine and combines the local estimators by taking the average.
	\item Robust estimator with Lasso (REL) on a single machine with pooled data (see Eq. \eqref{eq:pool}), which is denoted by pooled REL.
\end{enumerate}
Note that the $\ell_1$-regularized QR estimator in \eqref{eq:QR} and the de-biased averaging divide-and-conquer estimator (see \cite{zhao2014general}) are not included in most comparisons because they are computationally very expensive to be implemented in our setting, with large $n$ and $p$.  Moreover, the de-biased estimator generates  a dense estimated coefficient due to the de-biasing procedure. In the experiment on computation efficiency, we compare the running time of our method to the $\ell_1$-regularized QR estimator. The result shows that our method achieves a similar performance as the $\ell_1$-regularized QR estimator and it is computationally much more efficient. 

\subsection{Effect of the Number of Iterations}
We first show the performance of our distribute REL by varying the number of iterations. We fix the sample size $n=10000$, local sample size $m=500$, the sparsity level $s=20$ and dimension $p=500$. We plot the $\ell_2$-error from the true QR coefficients versus the number of iterations. Since the Avg-DC only requires one-shot communication, we use a horizontal line to show its performance. The results are shown in Figure \ref{iter}. 
\begin{figure}[!ht]
	\centering
	\addtolength{\leftskip} {-4cm}
	\addtolength{\rightskip}{-4cm}
	\subfloat[Normal noise]{
		\includegraphics[width=\figspace\textwidth]{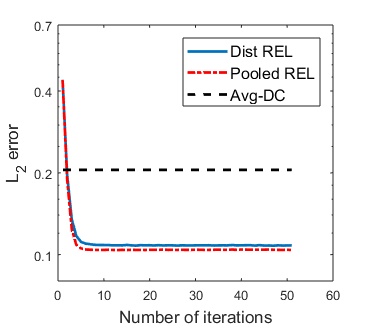}
		\label{iter_normal}}
	\hspace{-1.5em}
	\subfloat[Cauchy noise]{
		\includegraphics[width=\figspace\textwidth]{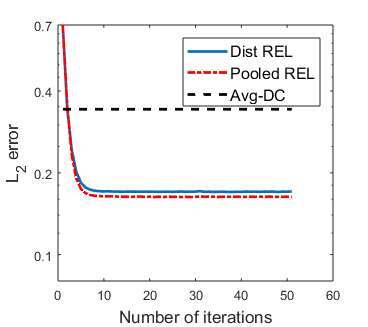}
		\label{iter_cauchy}}		
	\hspace{-1.5em}
	\subfloat[Exponential noise]{
		\includegraphics[width=\figspace\textwidth]{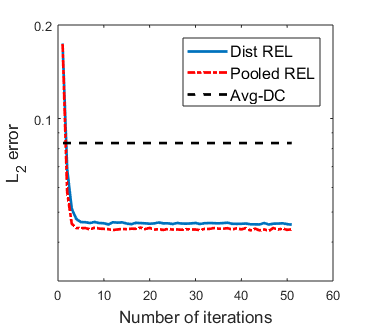}
		\label{iter_exp}}
	\caption{The $\ell_2$-error from the true QR coefficient versus the number of iterations. The sample size $n$ is fixed to $n=10000$ and the local sample size $m$ is 500.}\label{iter}
\end{figure}
From the result, both pooled REL and distributed REL outperform the Avg-DC algorithm and become stable after a few iterations.  Therefore, for the rest of the experiments in this section, we use 50 as the number of iterations in the algorithm. Moreover, the distributed REL almost matches the performance of pooled REL for all three noises.

\subsection{Effect of the QR Loss Under Heavy-Tailed Noise}
We study the effect of the QR loss in the presence of heavy-tailed noise. We compare with the standard Lasso estimator in a single machine setting with pooled data. We vary the sample size $n$ and compute the $F_1$-score and the $\ell_2$-error for the distributed REL, Pooled REL, Avg-DC, and the Lasso estimator.  The $F_1$-score is defined as 
\[
F_{1}=\left({\frac {\mathrm {recall} ^{-1}+\mathrm {precision} ^{-1}}{2}}\right)^{-1}=2\cdot {\frac {\mathrm {precision} \cdot \mathrm {recall} }{\mathrm {precision} +\mathrm {recall} }},
\]
which is commonly used as an evaluation of support recovery (note that $F_1$-score=1 implies perfect support recovery). In Table \ref{robust_normal}, \ref{robust_cauchy} and \ref{robust_exp}, we report the results for all three types of noises.
\begin{table}
	\caption{The $F_1$-score and $\ell_2$-error of the distributed REL, pooled REL, Avg-DC, and Lasso estimator under different sample size $n$. Noises are generated from normal distribution. The local sample size is fixed to $m=500$.\label{robust_normal}}
	\centering
	\resizebox{\textwidth}{!}{%
		\begin{tabular}{c|cc|cc|cc|cc}
			\hline
			\multirow{2}{*}{$n$} & \multicolumn{2}{c|}{Dist REL} & \multicolumn{2}{c|}{Pooled REL} & \multicolumn{2}{c|}{Avg-DC} & \multicolumn{2}{c}{Lasso} \\ \cline{2-9} 
			& $F_1$-score & $\ell_2$-error & $F_1$-score & $\ell_2$-error & $F_1$-score & $\ell_2$-error & $F_1$-score & $\ell_2$-error \\ \hline
			2500 & 0.90 & 0.189 & 0.83 & 0.183 & 0.23 & 0.255 & 1.00 & 0.161 \\
			5000 & 0.95 & 0.138 & 0.91 & 0.132 & 0.14 & 0.221 & 1.00 & 0.113 \\
			10000 & 0.97 & 0.102 & 0.93 & 0.097 & 0.10 & 0.203 & 1.00 & 0.079 \\
			15000 & 0.98 & 0.085 & 0.96 & 0.083 & 0.09 & 0.196 & 1.00 & 0.065 \\
			20000 & 0.99 & 0.073 & 0.96 & 0.069 & 0.08 & 0.192 & 1.00 & 0.056 \\
			25000 & 0.99 & 0.067 & 0.97 & 0.050 & 0.08 & 0.196 & 1.00 & 0.050 \\ \hline
		\end{tabular}%
	}
\end{table}
\begin{table}
	\caption{The $F_1$-score and $\ell_2$-error of the distributed REL, pooled REL, Avg-DC, and Lasso estimator under different sample size $n$. Noises are generated from Cauchy distribution. The local sample size is fixed to $m=500$.\label{robust_cauchy}}
	\centering
	\resizebox{\textwidth}{!}{%
		\begin{tabular}{c|cc|cc|cc|cc}
			\hline
			\multirow{2}{*}{$n$} & \multicolumn{2}{c|}{Dist REL} & \multicolumn{2}{c|}{Pooled REL} & \multicolumn{2}{c|}{Avg-DC} & \multicolumn{2}{c}{Lasso} \\ \cline{2-9} 
			& $F_1$-score & $\ell_2$-error & $F_1$-score & $\ell_2$-error & $F_1$-score & $\ell_2$-error & $F_1$-score & $\ell_2$-error \\ \hline
			2500 & 0.84 & 0.320 & 0.75 & 0.312 & 0.25 & 0.436 & 0.25 & 151.4 \\
			5000 & 0.92 & 0.229 & 0.85 & 0.221 & 0.16 & 0.380 & 0.26 & 138.8 \\
			10000 & 0.96 & 0.168 & 0.89 & 0.160 & 0.11 & 0.349 & 0.27 & 128.3 \\
			15000 & 0.98 & 0.139 & 0.92 & 0.132 & 0.09 & 0.338 & 0.25 & 132.1 \\
			20000 & 0.97 & 0.118 & 0.93 & 0.113 & 0.08 & 0.329 & 0.26 & 121.0 \\
			25000 & 0.98 & 0.107 & 0.94 & 0.101 & 0.08 & 0.330 & 0.23 & 120.8 \\ \hline
		\end{tabular}%
	}
\end{table}
\begin{table}
	\caption{The $F_1$-score and $\ell_2$-error of the distributed REL, pooled REL, Avg-DC, and Lasso estimator under different sample size $n$. Noises are generated from exponential distribution. The local sample size is fixed to $m=500$.\label{robust_exp}}
	\centering
	\resizebox{\textwidth}{!}{%
		\begin{tabular}{c|cc|cc|cc|cc}
			\hline
			\multirow{2}{*}{$n$} & \multicolumn{2}{c|}{Dist REL} & \multicolumn{2}{c|}{Pooled REL} & \multicolumn{2}{c|}{Avg-DC} & \multicolumn{2}{c}{Lasso} \\ \cline{2-9} 
			& $F_1$-score & $\ell_2$-error & $F_1$-score & $\ell_2$-error & $F_1$-score & $\ell_2$-error & $F_1$-score & $\ell_2$-error \\ \hline
			2500 & 0.96 & 0.093 & 0.91 & 0.089 & 0.25 & 0.115 & 1.00 & 0.102 \\
			5000 & 0.98 & 0.069 & 0.92 & 0.066 & 0.15 & 0.101 & 1.00 & 0.094 \\
			10000 & 0.99 & 0.051 & 0.96 & 0.048 & 0.10 & 0.092 & 1.00 & 0.069 \\
			15000 & 0.99 & 0.043 & 0.97 & 0.040 & 0.09 & 0.089 & 1.00 & 0.054 \\
			20000 & 1.00 & 0.037 & 0.98 & 0.034 & 0.08 & 0.086 & 1.00 & 0.048 \\
			25000 & 0.99 & 0.033 & 0.98 & 0.031 & 0.08 & 0.087 & 1.00 & 0.043 \\ \hline
		\end{tabular}%
	}
\end{table}
As expected, when the noise is normal, the Lasso estimator has smaller $\ell_2$-error and better support recovery. However, when the noise has a slightly heavier tail (e.g., exponential noise), both the distributed REL and pooled REL outperform the Lasso estimator in $\ell_2$-error. In the case of heavy-tailed noise (e.g., Cauchy noise), the Lasso approach completely fails with very large $\ell_2$-errors while the distributed REL is much better in both $\ell_2$-error and support recovery. It is clear that the Lasso estimator is not robust to heavy-tailed noises, and therefore we omit the Lasso estimator in the rest of the simulation studies.

Another interesting phenomena revealed in Tables \ref{robust_normal}-\ref{robust_exp} is that, in terms of the $F_1$-score, the distributed REL is slightly better than pooled REL. This is indeed affected by the selection of regularization parameter $\lambda_n$. According to our Theorem \ref{thm:betainf}, we set $\lambda_n$ for the first round on the order of $s \log n/m$, where $m$ is the local sample size and $n$ the total sample size. For the pooled estimator where $m=n$, this term becomes $s \log n/n$, which becomes smaller. Therefore, our distributed estimator has already eliminated many features for the first round due to a larger regularization parameter, which leads to a slightly better precision. It is noted that this also happens in the following experiments.

\subsection{Effect of Sample Size and Local Sample Size}
\label{sec:effect}
In this section, we investigate how the performance of the distributed REL changes with the total sample size $n$ and the local sample size $m$. We also compare our estimator with the Communication-efficient Surrogate Likelihood (CSL) estimator proposed in \cite{jordan2018communication}. The original method in \cite{jordan2018communication} requires  second-order differentiable loss functions, which is not directly applicable to quantile loss function. Thus, we adopt a smoothing technique to smooth the QR loss function as in \cite{horowitz1998bootstrap, chen2019}. We fix sparsity level $s=20$, $p=500$, and vary the sample size $n\in\{5000,10000,20000\}$ and the local sample size $m\in\{200,500,1000\}$. The precision, recall of the support recovery and the $\ell_2$-error are reported for each estimator. The results are shown in Table \ref{mn_normal}, \ref{mn_cauchy} and \ref{mn_exp}.
\begin{table}
	\caption{The $\ell_2$-error, precision, and recall of the three estimators under different combinations of the sample size $n$ and local sample size $m$. Noises are generated from normal distribution.\label{mn_normal}}
	\centering
	\resizebox{\textwidth}{!}{%
		\begin{tabular}{c|c|ccc|ccc|ccc}
			\hline
			\multicolumn{2}{c|}{$m$} & \multicolumn{3}{c|}{200} & \multicolumn{3}{c|}{500} & \multicolumn{3}{c}{1000} \\ \hline
			\multicolumn{2}{c|}{$n$} & 5000 & 10000 & 20000 & 5000 & 10000 & 20000 & 5000 & 10000 & 20000 \\ \hline
			\multirow{3}{*}{\begin{tabular}[c]{@{}c@{}}Pooled\\ REL\end{tabular}} & Precision & 0.79 & 0.85 & 0.92 & 0.79 & 0.89 & 0.93 & 0.78 & 0.85 & 0.92 \\
			& Recall & 1.00 & 1.00 & 1.00 & 1.00 & 1.00 & 1.00 & 1.00 & 1.00 & 1.00 \\
			& $\ell_2$-error & 0.136 & 0.098 & 0.071 & 0.138 & 0.101 & 0.073 & 0.135 & 0.100 & 0.072 \\ \hline
			\multirow{3}{*}{\begin{tabular}[c]{@{}c@{}}Dist\\ REL\end{tabular}} & Precision & 0.98 & 0.99 & 1.00 & 0.91 & 0.95 & 0.98 & 0.83 & 0.89 & 0.95 \\
			& Recall & 1.00 & 1.00 & 1.00 & 1.00 & 1.00 & 1.00 & 1.00 & 1.00 & 1.00 \\
			& $\ell_2$-error & 0.154 & 0.111 & 0.081 & 0.142 & 0.105 & 0.076 & 0.137 & 0.102 & 0.074 \\ \hline
			\multirow{3}{*}{\begin{tabular}[c]{@{}c@{}}Avg\\ DC\end{tabular}} & Precision & 0.05 & 0.04 & 0.04 & 0.08 & 0.06 & 0.05 & 0.13 & 0.08 & 0.06 \\
			& Recall & 1.00 & 1.00 & 1.00 & 1.00 & 1.00 & 1.00 & 1.00 & 1.00 & 1.00 \\
			& $\ell_2$-error & 0.348 & 0.328 & 0.314 & 0.225 & 0.205 & 0.199 & 0.180 & 0.156 & 0.145 \\ \hline
			\multirow{3}{*}{CSL} & Precision & 0.86 & 0.85 & 0.88 & 0.08 & 1.00 & 1.00 & 1.00 & 1.00& 1.00 \\
			& Recall & 0.95 & 0.93 & 0.94 & 1.00 & 1.00 & 1.00 & 1.00 & 1.00 & 1.00 \\
			& $\ell_2$-error & 0.480 & 0.455 & 0.452 & 0.218 & 0.201 & 0.190 & 0.154 & 0.141 & 0.098 \\ \hline
		\end{tabular}
	}
\end{table}
\begin{table}
	\caption{The $\ell_2$-error, precision, and recall of the three estimators under different combinations of the sample size $n$ and local sample size $m$. Noises are generated from Cauchy distribution.\label{mn_cauchy}}
	\centering
	\resizebox{\textwidth}{!}{%
		\begin{tabular}{c|c|ccc|ccc|ccc}
			\hline
			\multicolumn{2}{c|}{$m$} & \multicolumn{3}{c|}{200} & \multicolumn{3}{c|}{500} & \multicolumn{3}{c}{1000} \\ \hline
			\multicolumn{2}{c|}{$n$} & 5000 & 10000 & 20000 & 5000 & 10000 & 20000 & 5000 & 10000 & 20000 \\ \hline
			\multirow{3}{*}{\begin{tabular}[c]{@{}c@{}}Pooled\\ REL\end{tabular}} & Precision & 0.72 & 0.84 & 0.89 & 0.75 & 0.82 & 0.88 & 0.70 & 0.81 & 0.87 \\
			& Recall & 1.00 & 1.00 & 1.00 & 1.00 & 1.00 & 1.00 & 1.00 & 1.00 & 1.00 \\
			& $\ell_2$-error & 0.220 & 0.159 & 0.118 & 0.221 & 0.161 & 0.116 & 0.221 & 0.156 & 0.114 \\ \hline
			\multirow{3}{*}{\begin{tabular}[c]{@{}c@{}}Dist\\ REL\end{tabular}} & Precision & 0.98 & 0.99 & 1.00 & 0.86 & 0.91 & 0.95 & 0.76 & 0.87 & 0.92 \\
			& Recall & 1.00 & 1.00 & 1.00 & 1.00 & 1.00 & 1.00 & 1.00 & 1.00 & 1.00 \\
			& $\ell_2$-error & 0.251 & 0.181 & 0.134 & 0.230 & 0.169 & 0.122 & 0.223 & 0.158 & 0.117 \\ \hline
			\multirow{3}{*}{\begin{tabular}[c]{@{}c@{}}Avg\\ DC\end{tabular}} & Precision & 0.05 & 0.04 & 0.04 & 0.08 & 0.06 & 0.04 & 0.14 & 0.08 & 0.06 \\
			& Recall & 1.00 & 1.00 & 1.00 & 1.00 & 1.00 & 1.00 & 1.00 & 1.00 & 1.00 \\
			& $\ell_2$-error & 0.704 & 0.671 & 0.667 & 0.375 & 0.355 & 0.332 & 0.291 & 0.245 & 0.235 \\ \hline
			\multirow{3}{*}{CSL} & Precision & 0.09 & 0.12 & 0.12 & 0.28 & 0.35 & 0.48 & 0.64 & 0.77& 0.89 \\
			& Recall & 0.91 & 0.93 & 0.90 & 0.97 & 0.97 & 0.98 & 0.98 & 0.98 & 0.99 \\
			& $\ell_2$-error & 0.834 & 0.790 & 0.728 & 0.324 & 0.327 & 0.312 & 0.255 & 0.195 & 0.171 \\ \hline
		\end{tabular}
	}
\end{table}
\begin{table}
	\caption{The $\ell_2$-error, precision, and recall of the three estimators under different combinations of the sample size $n$ and local sample size $m$. Noises are generated from exponential distribution.\label{mn_exp}}
	\centering
	\resizebox{\textwidth}{!}{%
		\begin{tabular}{c|c|ccc|ccc|ccc}
			\hline
			\multicolumn{2}{c|}{$m$} & \multicolumn{3}{c|}{200} & \multicolumn{3}{c|}{500} & \multicolumn{3}{c}{1000} \\ \hline
			\multicolumn{2}{c|}{$n$} & 5000 & 10000 & 20000 & 5000 & 10000 & 20000 & 5000 & 10000 & 20000 \\ \hline
			\multirow{3}{*}{\begin{tabular}[c]{@{}c@{}}Pooled\\ REL\end{tabular}} & Precision & 0.90 & 0.98 & 0.98 & 0.88 & 0.94 & 0.96 & 0.86 & 0.93 & 0.96 \\
			& Recall & 1.00 & 1.00 & 1.00 & 1.00 & 1.00 & 1.00 & 1.00 & 1.00 & 1.00 \\
			& $\ell_2$-error & 0.060 & 0.045 & 0.031 & 0.059 & 0.042 & 0.032 & 0.059 & 0.042 & 0.030 \\ \hline
			\multirow{3}{*}{\begin{tabular}[c]{@{}c@{}}Dist\\ REL\end{tabular}} & Precision & 1.00 & 1.00 & 1.00 & 0.95 & 0.98 & 0.99 & 0.91 & 0.95 & 0.98 \\
			& Recall & 1.00 & 1.00 & 1.00 & 1.00 & 1.00 & 1.00 & 1.00 & 1.00 & 1.00 \\
			& $\ell_2$-error & 0.076 & 0.061 & 0.043 & 0.062 & 0.044 & 0.034 & 0.060 & 0.042 & 0.031 \\ \hline
			\multirow{3}{*}{\begin{tabular}[c]{@{}c@{}}Avg\\ DC\end{tabular}} & Precision & 0.05 & 0.04 & 0.04 & 0.07 & 0.06 & 0.04 & 0.15 & 0.09 & 0.05 \\
			& Recall & 1.00 & 1.00 & 1.00 & 1.00 & 1.00 & 1.00 & 1.00 & 1.00 & 1.00 \\
			& $\ell_2$-error & 0.168 & 0.162 & 0.154 & 0.090 & 0.084 & 0.079 & 0.072 & 0.062 & 0.054 \\ \hline
			\multirow{3}{*}{CSL} & Precision & 0.86 & 0.85 & 0.88 & 0.08 & 1.00 & 1.00 & 1.00 & 1.00& 1.00 \\
			& Recall & 0.95 & 0.93 & 0.94 & 1.00 & 1.00 & 1.00 & 1.00 & 1.00 & 1.00 \\
			& $\ell_2$-error & 0.480 & 0.455 & 0.452 & 0.218 & 0.201 & 0.190 & 0.154 & 0.141 & 0.098 \\ \hline
		\end{tabular}
	}
\end{table}

\begin{figure}[p]
	\centering
	\addtolength{\leftskip} {-4cm}
	\addtolength{\rightskip}{-4cm}
	\subfloat[Normal noise]{
		\includegraphics[width=\figspace\textwidth]{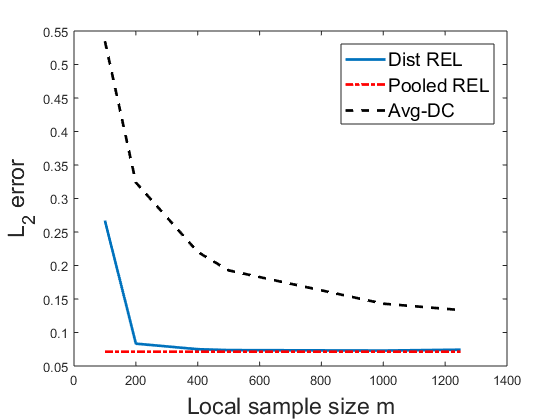}}
	\hspace{-1.5em}
	\subfloat[Cauchy noise]{
		\includegraphics[width=\figspace\textwidth]{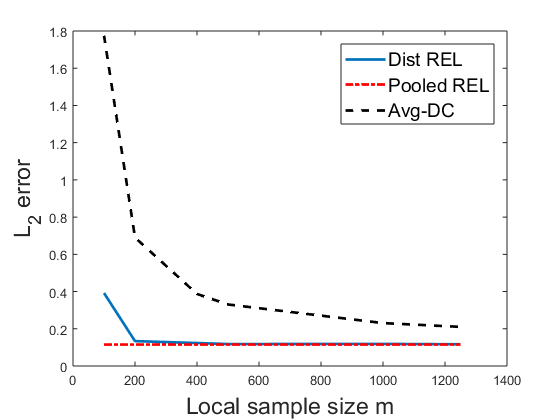}}
	\hspace{-1.5em}
	\subfloat[Exponential noise]{
		\includegraphics[width=\figspace\textwidth]{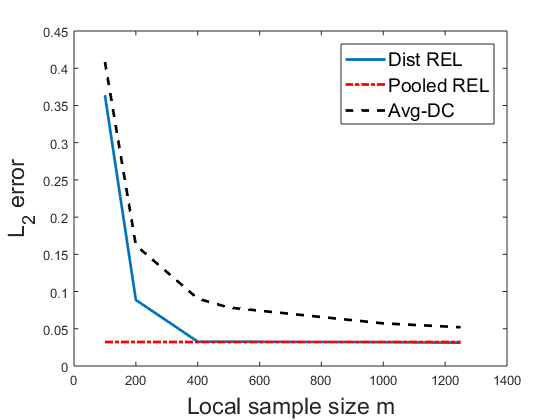}}
	\caption{The $\ell_2$-error from the true QR coefficient versus the local sample size $m$, with the total sample size fixed to $n = 20000$.}\label{m}
\end{figure}
\begin{figure}[p]
	\centering
	\addtolength{\leftskip} {-4cm}
	\addtolength{\rightskip}{-4cm}
	\subfloat[Normal noise]{
		\includegraphics[width=\figspace\textwidth]{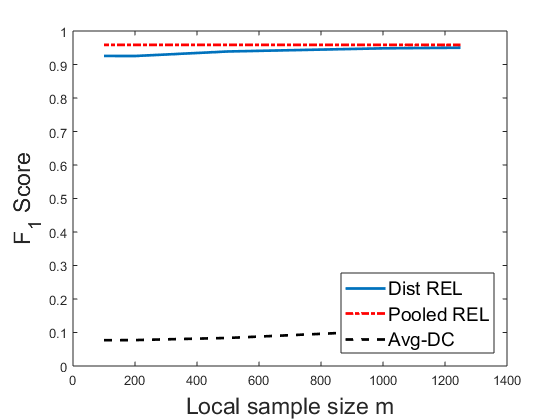}}
	\hspace{-1.5em}
	\subfloat[Cauchy noise]{
		\includegraphics[width=\figspace\textwidth]{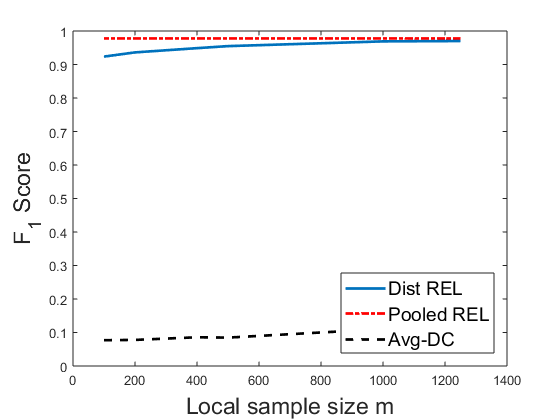}}
	\hspace{-1.5em}
	\subfloat[Exponential noise]{
		\includegraphics[width=\figspace\textwidth]{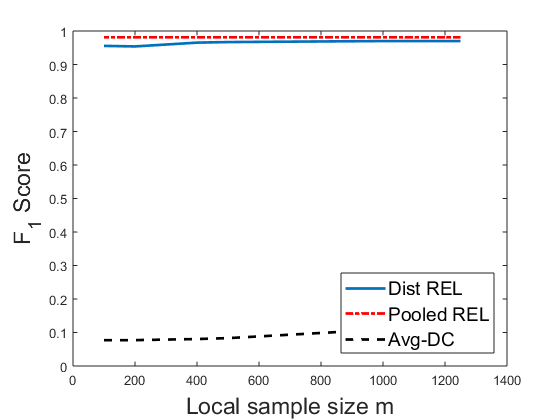}}
	\caption{The $F_1$-score versus the local sample size $m$, with the total sample size fixed to $n = 20000$.}\label{m_f1}
\end{figure}

From the results, we observe that both distributed REL and pooled REL outperform the Avg-DC algorithm and CSL estimator in all settings. The $\ell_2$-error of the distributed REL improves as the local sample size $m$ grows and it becomes close to pooled REL when $m$ is large. This is expected since the pooled REL is a special case of distributed REL with $m=n$. We also observe that the precision and recall of the distributed REL are both close to 1, which indicates good support recovery. In particular, the recall of our distributed REL is always 1, implying that all the relevant variables are selected. The precision of our method is close to 1, which indicates that only a very small number of irrelevant variables are selected. On the other hand, the precision of Avg-DC is very small because the averaging procedure results in a dense estimator, especially when $m$ is small. In addition, the performance of CSL estimator heavily depends on $m$. For example, for Cauchy error distribution in Table \ref{mn_cauchy}, a smaller $m$ leads to a relatively poor performance. 

For better visualization, with the sample size $n=20000$ fixed, we vary the local sample size $m$ and plot the $\ell_2$-error and $F_1$-score of pooled REL, distributed REL and Avg-DC estimator.
The results are presented in Figure \ref{m} and \ref{m_f1}. Similarly, in Figure \ref{n} and \ref{n_f1}, we fix the local sample size $m=500$ and vary the total sample size $n$.

\begin{figure}[!t]
	\centering
	\addtolength{\leftskip} {-4cm}
	\addtolength{\rightskip}{-4cm}
	\subfloat[Normal noise]{
		\includegraphics[width=\figspace\textwidth]{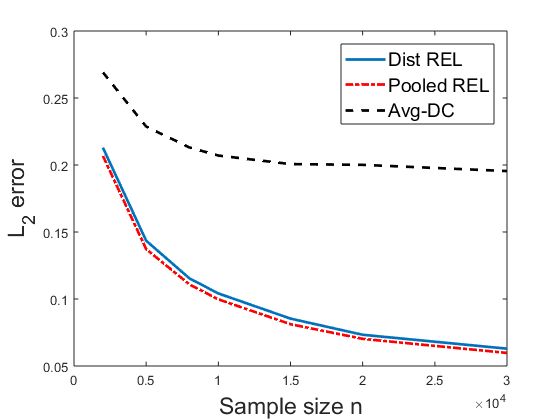}}
	\hspace{-1.5em}
	\subfloat[Cauchy noise]{
		\includegraphics[width=\figspace\textwidth]{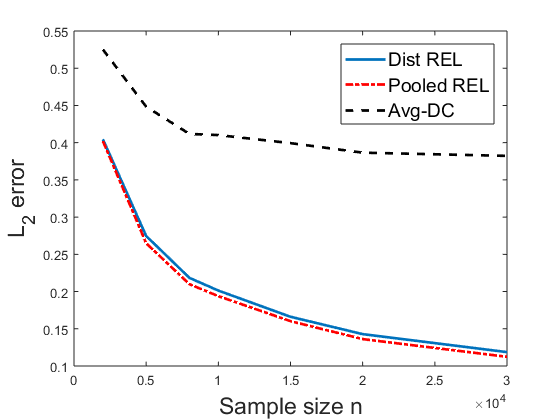}}
	\hspace{-1.5em}
	\subfloat[Exponential noise]{
		\includegraphics[width=\figspace\textwidth]{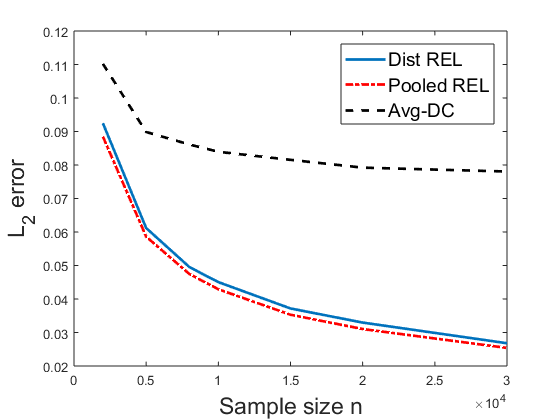}}
	\caption{The $\ell_2$-error from the true QR coefficient versus the sample size $n$, with the local sample size fixed to $m=500$.}\label{n}
\end{figure}
\begin{figure}[!t]
	\centering
	\addtolength{\leftskip} {-4cm}
	\addtolength{\rightskip}{-4cm}
	\subfloat[Normal noise]{
		\includegraphics[width=\figspace\textwidth]{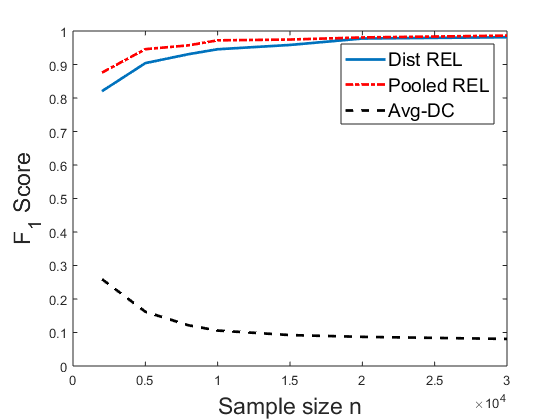}}
	\hspace{-1.5em}
	\subfloat[Cauchy noise]{
		\includegraphics[width=\figspace\textwidth]{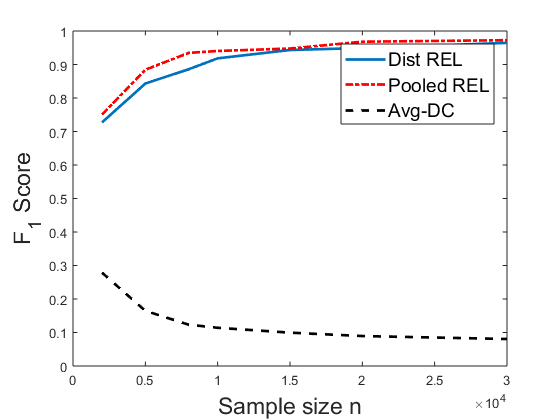}}
	\hspace{-1.5em}
	\subfloat[Exponential noise]{
		\includegraphics[width=\figspace\textwidth]{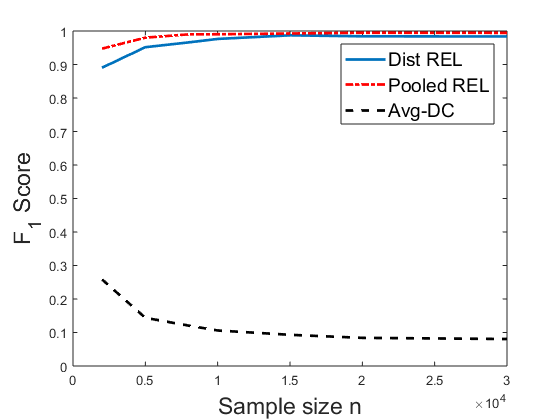}}
	\caption{The $F_1$-score versus the sample size $n$, with the local sample size fixed to $m=500$.}\label{n_f1}
\end{figure}
From Figure \ref{m} we can see that the $\ell_2$-error of distributed REL is close to that of pooled REL when $m$ is not too small, and both of them outperform the Avg-DC estimator. From Figure \ref{n} we observe that the $\ell_2$-error of distributed REL is close to that of pooled REL and both errors decrease as the sample size $n$ becomes large. However, the $\ell_2$-error of the Avg-DC estimator stays large and fails to converge as the sample size $n$ increases. From Figure \ref{m_f1} and \ref{n_f1} we can see that the $F_1$-score of both distributed REL and pooled REL are close to 1, while the Avg-DC approach clearly fails in support recovery in high-dimensional settings.

\subsection{Sensitivity Analysis for the Bandwidth}\label{sec:sensitivity}
In this section, we study the sensitivity of the scaling constant in the bandwidth of the proposed REL. Recall that the bandwidth is $h = ca_{n,g}$ where $a_{n,g}$ is defined in \eqref{eq:a} with $c>0$ being the scaling constant. We vary the sample size $n$ and the constant $c$ from $0.5$ to 10 and compute the $F_1$-score and the $\ell_2$-error of the distributed REL, pooled REL, and the Avg-DC estimator. Due to space limitations, we report the Cauchy noise case as an example. For other noises, the performance is even less sensitive. The results are shown in Table \ref{bandwidth}. 
\begin{table}
	\caption{The $F_1$-score and $\ell_2$-error of the distributed REL, pooled REL, and Avg-DC  under different sample size $n$ and choices of bandwidth constant $c$. Local sample size $m=500$. Noises are generated from Cauchy distribution.\label{bandwidth}}
	\centering
	\resizebox{0.85\textwidth}{!}{%
		\begin{tabular}{c|c|cc|cc|cc}
			\hline
			\multirow{2}{*}{$n$} & \multirow{2}{*}{$c$} & \multicolumn{2}{c|}{Dist REL} & \multicolumn{2}{c|}{Pooled REL} & \multicolumn{2}{c}{Avg-DC} \\ \cline{3-8} 
			&  & $F_1$-score & $\ell_2$-error & $F_1$-score & $\ell_2$-error & $F_1$-score & $\ell_2$-error \\ \hline
			5000 & 0.5 & 0.99 & 0.249 & 0.96 & 0.236 & 0.17 & 0.377 \\
			10000 & 0.5 & 1.00 & 0.183 & 0.99 & 0.171 & 0.12 & 0.356 \\
			20000 & 0.5 & 0.99 & 0.130 & 0.99 & 0.123 & 0.09 & 0.348 \\ \hline
			5000 & 1 & 0.99 & 0.253 & 0.96 & 0.241 & 0.16 & 0.373 \\
			10000 & 1 & 0.99 & 0.179 & 0.98 & 0.170 & 0.11 & 0.345 \\
			20000 & 1 & 1.00 & 0.125 & 0.98 & 0.117 & 0.09 & 0.328 \\ \hline
			5000 & 2 & 0.99 & 0.259 & 0.97 & 0.245 & 0.16 & 0.38 \\
			10000 & 2 & 1.00 & 0.188 & 0.98 & 0.177 & 0.11 & 0.347 \\
			20000 & 2 & 1.00 & 0.131 & 0.99 & 0.124 & 0.09 & 0.332 \\ \hline
			5000 & 5 & 0.99 & 0.255 & 0.97 & 0.239 & 0.16 & 0.378 \\
			10000 & 5 & 1.00 & 0.185 & 0.98 & 0.173 & 0.11 & 0.349 \\
			20000 & 5 & 1.00 & 0.138 & 0.98 & 0.124 & 0.09 & 0.339 \\ \hline
			5000 & 10 & 1.00 & 0.270 & 0.99 & 0.252 & 0.16 & 0.382 \\
			10000 & 10 & 1.00 & 0.194 & 0.99 & 0.180 & 0.1 & 0.346 \\
			20000 & 10 & 1.00 & 0.136 & 0.98 & 0.121 & 0.09 & 0.331 \\ \hline
		\end{tabular}
	}
\end{table}

From Table \ref{bandwidth}, we observe that both distributed REL and pooled REL exhibit good performance under all choices of bandwidth constant. Therefore even under a suboptimal choice of bandwidth constant, the distributed REL still achieves small $\ell_2$-error and good support recovery.

\subsection{Effect of the Sparsity}
In this section we investigate how the performance of the distributed REL algorithm  changes with the sparsity level of the true coefficient $\be^*$. We fix the sample size $n = 10000$ and the local sample size $m = 500$, and we set the constant $c_0$ in $h_g$ to be 0.01. Recall that the true coefficient is set to be $$\be^* = (\frac{10}{s},\frac{20}{s},\frac{30}{s},\ldots,\frac{10(s-1)}{s},10,0,0\ldots,0).$$ We vary the sparsity level $s$ in $\{5,10,20,30,50,100\}$ and report the precision, recall and $\ell_2$-error. Since the $\ell_2$-norm of the true coefficient $\be^*$ changes with the sparsity level $s$, we also report the relative $\ell_2$-error which is defined by $|\widehat{\be}-\be^*|_2/|\be^*|_2$. The results are shown in Table \ref{s_homo}, \ref{s_cauchy} and \ref{s_exp}.\\
From the result, we can observe that the $\ell_2$-errors of all three estimators become larger as the sparsity level $s$ increases and the distributed REL algorithm performs much better than the Avg-DC algorithm. Moreover, the performance of the distributed REL is very close to the performance of the pooled REL.
\begin{table}
	\caption{The $\ell_2$-error, precision, and recall of the three estimators with different sparsity level $s$. Noises are generated from normal distribution. The local sample size is fixed to $m=500$.\label{s_homo}}
	\centering
	\begin{tabular}{c|c|cccccc}
		\hline
		\multicolumn{2}{c|}{Sparsity $s$} & 5 & 10 & 20 & 30 & 50 & 100 \\ \hline
		\multirow{4}{*}{\begin{tabular}[c]{@{}c@{}}Pooled\\ REL\end{tabular}} & Precision & 0.98 & 0.96 & 0.86 & 0.82 & 0.73 & 0.66 \\
		& Recall & 1.00 & 1.00 & 1.00 & 1.00 & 1.00 & 1.00 \\
		& $\ell_2$-error & 0.063 & 0.080 & 0.096 & 0.117 & 0.141 & 0.191 \\
		& Relative $\ell_2$-error($\times 10^{-2}$) & 0.426 & 0.408 & 0.360 & 0.361 & 0.341 & 0.329 \\ \hline
		\multirow{4}{*}{\begin{tabular}[c]{@{}c@{}}Dist\\ REL\end{tabular}} & Precision & 1.00 & 0.98 & 0.94 & 0.93 & 0.91 & 0.88 \\
		& Recall & 1.00 & 1.00 & 1.00 & 1.00 & 1.00 & 1.00 \\
		& $\ell_2$-error & 0.065 & 0.082 & 0.101 & 0.123 & 0.150 & 0.202 \\
		& Relative $\ell_2$-error($\times 10^{-2}$) & 0.441 & 0.418 & 0.379 & 0.379 & 0.363 & 0.347 \\ \hline
		\multirow{4}{*}{\begin{tabular}[c]{@{}c@{}}Avg\\ DC\end{tabular}} & Precision & 0.02 & 0.03 & 0.06 & 0.08 & 0.11 & 0.20 \\
		& Recall & 1.00 & 1.00 & 1.00 & 1.00 & 1.00 & 1.00 \\
		& $\ell_2$-error & 0.147 & 0.175 & 0.204 & 0.243 & 0.280 & 0.368 \\
		& Relative $\ell_2$-error($\times 10^{-2}$) & 0.988 & 0.890 & 0.760 & 0.751 & 0.675 & 0.633 \\ \hline
	\end{tabular}
\end{table}
\begin{table}
	\caption{The $\ell_2$-error, precision, and recall of the three estimators with different sparsity level $s$. Noises are generated from Cauchy distribution. The local sample size is fixed to $m=500$.\label{s_cauchy}}
	\centering
	\begin{tabular}{c|c|cccccc}
		\hline
		\multicolumn{2}{c|}{Sparsity $s$} & 5 & 10 & 20 & 30 & 50 & 100 \\ \hline
		\multirow{4}{*}{\begin{tabular}[c]{@{}c@{}}Pool\\ QR\end{tabular}} & Precision & 0.95 & 0.91 & 0.79 & 0.73 & 0.66 & 0.64 \\
		& Recall & 1.00 & 1.00 & 1.00 & 1.00 & 1.00 & 1.00 \\
		& $\ell_2$-error & 0.103 & 0.129 & 0.156 & 0.186 & 0.230 & 0.318 \\
		& Relative $\ell_2$-error($\times 10^{-2}$) & 0.696 & 0.656 & 0.581 & 0.574 & 0.555 & 0.547 \\ \hline
		\multirow{4}{*}{\begin{tabular}[c]{@{}c@{}}Dist\\ QR\end{tabular}} & Precision & 0.97 & 0.95 & 0.91 & 0.87 & 0.86 & 0.84 \\
		& Recall & 1.00 & 1.00 & 1.00 & 1.00 & 1.00 & 1.00 \\
		& $\ell_2$-error & 0.105 & 0.132 & 0.163 & 0.194 & 0.239 & 0.330 \\
		& Relative $\ell_2$-error($\times 10^{-2}$) & 0.709 & 0.674 & 0.608 & 0.598 & 0.578 & 0.567 \\ \hline
		\multirow{4}{*}{\begin{tabular}[c]{@{}c@{}}Avg\\ DC\end{tabular}} & Precision & 0.02 & 0.04 & 0.06 & 0.07 & 0.11 & 0.20 \\
		& Recall & 1.00 & 1.00 & 1.00 & 1.00 & 1.00 & 1.00 \\
		& $\ell_2$-error & 0.264 & 0.319 & 0.347 & 0.419 & 0.542 & 0.885 \\
		& Relative $\ell_2$-error($\times 10^{-2}$) & 1.779 & 1.628 & 1.295 & 1.293 & 1.308 & 1.252 \\ \hline
	\end{tabular}
\end{table}
\begin{table}
	\caption{The $\ell_2$-error, precision, and recall of the three estimators with different sparsity level $s$. Noises are generated from exponential distribution. The local sample size is fixed to $m=500$.\label{s_exp}}
	\centering
	\begin{tabular}{c|c|cccccc}
		\hline
		\multicolumn{2}{c|}{Sparsity $s$} & 5 & 10 & 20 & 30 & 50 & 100 \\ \hline
		\multirow{4}{*}{\begin{tabular}[c]{@{}c@{}}Pooled\\ REL\end{tabular}} & Precision & 0.97 & 0.97 & 0.95 & 0.92 & 0.87 & 0.79 \\
		& Recall & 1.00 & 1.00 & 1.00 & 1.00 & 1.00 & 1.00 \\
		& $\ell_2$-error & 0.026 & 0.034 & 0.043 & 0.049 & 0.062 & 0.080 \\
		& Relative $\ell_2$-error($\times 10^{-2}$) & 0.178 & 0.171 & 0.160 & 0.151 & 0.149 & 0.138 \\ \hline
		\multirow{4}{*}{\begin{tabular}[c]{@{}c@{}}Dist\\ REL\end{tabular}} & Precision & 0.99 & 0.99 & 0.98 & 0.98 & 0.98 & 0.99 \\
		& Recall & 1.00 & 1.00 & 1.00 & 1.00 & 1.00 & 1.00 \\
		& $\ell_2$-error & 0.027 & 0.035 & 0.045 & 0.052 & 0.066 & 0.092 \\
		& Relative $\ell_2$-error($\times 10^{-2}$) & 0.185 & 0.180 & 0.169 & 0.161 & 0.160 & 0.158 \\ \hline
		\multirow{4}{*}{\begin{tabular}[c]{@{}c@{}}Avg\\ DC\end{tabular}} & Precision & 0.02 & 0.04 & 0.05 & 0.07 & 0.11 & 0.20 \\
		& Recall & 1.00 & 1.00 & 1.00 & 1.00 & 1.00 & 1.00 \\
		& $\ell_2$-error & 0.054 & 0.065 & 0.083 & 0.099 & 0.113 & 0.151 \\
		& Relative $\ell_2$-error($\times 10^{-2}$) & 0.365 & 0.329 & 0.311 & 0.305 & 0.273 & 0.260 \\ \hline
	\end{tabular}
\end{table}

\subsection{Computation Time Comparison}\label{sec:time}
We further study the computation efficiency of our proposed estimator. We fix the local sample size $m$, dimension $p$, and vary the sample size $n$. In Table \ref{time}, we report the $F_1$-score, $\ell_2$-error, and the computation time of  distributed REL, pooled REL, Avg-DC, and the $\ell_1$-regularized  QR estimator. To solve the $\ell_1$-regularized  QR estimator, we formulate it into a standard linear programming problem (LP) and solve it by Gurobi \citep{gurobi}, which is the state-of-the-art LP solver. We implement the three distributed algorithms (distributed REL, pooled REL and Avg-DC) in a fully synchronized distributed setting.
\begin{table}
	\caption{The $F_1$-score, $\ell_2$-error, and computation time of the distributed REL, pooled REL, Avg-DC, and $\ell_1$-regularized QR estimator under different sample size $n$. Noises are generated from Cauchy distribution. The local sample size is fixed to $m=500$.\label{time}}
	\centering
	\centering
	\begin{tabular}{c|ccc|ccc}
		\hline
		\multirow{2}{*}{$n$} & \multicolumn{3}{c|}{Dist REL} &\multicolumn{3}{c}{Pooled REL}  \\ \cline{2-7} 
		& $F_1$-score & $\ell_2$-error & Time & $F_1$-score & $\ell_2$-error & Time \\ \hline
		5000 & 0.95 & 0.137 & 0.40 &0.90 & 0.132 & 0.44\\
		10000 & 0.97 & 0.099 & 0.42 & 0.92 & 0.095 & 0.45\\
		15000 & 0.98 & 0.083 & 0.42 & 0.95 & 0.080 & 0.47 \\
		20000 & 0.99 & 0.074 & 0.44 & 0.96 & 0.071 & 0.48\\ 
		\hline
		\multirow{2}{*}{$n$} & \multicolumn{3}{c|}{Avg-DC} & \multicolumn{3}{c}{$\ell_1$-QR} \\ \cline{2-7} 
		& $F_1$-score & $\ell_2$-error & Time & $F_1$-score & $\ell_2$-error & Time \\ \hline
		5000 & 0.15 & 0.223 & 2.82 &  0.95&  0.132&  159.6\\
		10000 &  0.10 & 0.202 & 3.08 &  0.97&  0.091&  576.1\\
		15000 &  0.09 & 0.198 & 3.07 &  0.98&  0.077&  1223.1\\
		20000 & 0.08 & 0.192 & 3.15 &  0.99&  0.068&  2059.3\\ \hline
	\end{tabular}%
\end{table}

From Table \ref{time} we can see that the distributed REL is much faster than the $\ell_1$-regularized QR estimator. In fact, for larger sample size (i.e., $n> 20000$), we cannot implement the $\ell_1$-regularized QR method due to memory and computation time issues. We also note that the computation time of the pooled REL is similar to the distributed version. This is  because for the comparison propose, simulated datasets can still be fully stored in memory, and thus pooled REL takes the advantage of solving the entire optimization problem in memory. For large-scale datasets that cannot be stored in memory, the pool REL is no longer applicable.

\section{Conclusions and Future Directions}\label{sec:conclusion}

In this paper, we address the problem of distributed estimation for high-dimensional linear model with the presence of heavy-tailed noise. The proposed method achieves the same convergence rate as the ideal case  with pooled data. Furthermore, we establish the support recovery guarantee of the proposed method. One key insight from this work is that a non-smooth loss can be transformed into a smooth one by constructing a new response. Our method is essentially an iterative refinement approach in a distributed environment, which is superior to the averaging divide-and-conquer scheme.

One important future direction is to further investigate the inference problem. We note that \cite{zhao2014general} first provide the inference result based on averaging de-biased QR local estimators. As we mentioned, this approach might suffer from heavy computational cost and requires a condition on the number of machines. It would be interesting to develop computationally efficient inference approaches without any restriction on the number of machines. Moreover, the idea of transforming to $\ell_1$-regularized least-squares problem and the iterative distributed implementation can be generalized other high-dimensional problems, e.g., $\ell_1$-regularized Huber regression in robust statistics. Our algorithm can also be  generalized to handle other sparsity-inducing penalties, such as SCAD or MCP \citep{Fan-Li01,Zhang10}. Deriving the corresponding theoretical results for other sparsity-inducing penalties would be another interesting future direction.

\newpage
\appendix

\section*{Appendix}
The appendix is organized as follows. In Section \ref{sec:proofsupp}, we provide the proof of the main results and some technical lemmas. In Section \ref{sec:simsupp}, we provide additional simulation studies for distributed REL.

\section{Proof of Results}\label{sec:proofsupp}
In this section, we provide the proofs of our main results and  some technical lemmas. 
\subsection{Proof of Proposition \ref{prop:0}}
\begin{repproposition}{prop:0}
	Assume the following conditions hold
	\begin{eqnarray*}
		|\A\be^{*}-\b|_{\infty}\leq \lambda_{n}/2,
	\end{eqnarray*}
	\begin{eqnarray*}
		\min_{\delta: |\delta|_{1}\leq c_{1}\sqrt{s}|\delta|_{2}}\frac{\delta^{\mathrm{T}}\A\delta}{|\delta|^{2}_{2}}\geq c_{2},\quad c_{1},c_{2}>0.
	\end{eqnarray*}
	where $s$ is the sparsity of $\be^{*}$, i.e., $s=\sum_{j=0}^{p}\ind{\beta^{*}_{j}\neq 0}$. Then we have
	\begin{eqnarray*}
		|\hbe-\be^{*}|_{2}\leq c\sqrt{s}\lambda_{n}, 
	\end{eqnarray*}
	for some constant $c>0$.
\end{repproposition}
\begin{proof}
	We first show that $\abs{\hbe-\be^*}_1\leq 4\sqrt{s} \abs{\hbe-\be^*}_2$. Let $S$ be the support of $\be$. By the definition of $\hbe$, we have
	\begin{equation*}
		\begin{aligned}
			\frac{1}{2}\hbe^{\rm T} \A\hbe - \hbe^{\rm T} \b -(\frac{1}{2}\be^{*\rm T} \A\be^{*} - \be^{*\rm T} \b) \leq & \lambda_n(\abs{\be^*}_1-\abs{\hbe}_1)\\
			= & \lambda_n (\abs{\be^*_S}_1 - \abs{\hbe_S}_1-\abs{\hbe_{S^C}}_1)\\
			\leq & \lambda_n\abs{(\be^*-\hbe)_S}_1-\lambda_n\abs{(\be^*-\hbe)_{S^C}}_1.
		\end{aligned}
	\end{equation*}
	
	Since $\A$ is non-negative definite, we have 
	\begin{equation*}
		\begin{aligned}
			\frac{1}{2}\hbe^{\rm T} \A\hbe - \hbe^{\rm T} \b -(\frac{1}{2}\be^{*\rm T} \A\be^{*} - \be^{*\rm T} \b) \geq & (\A\be^*-\b)(\hbe-\be^*)\\
			\geq& -\abs{\A \be^*-\b}_\infty\abs{\hbe-\be^*}_1\\
			\geq& -\lambda_n\abs{\hbe-\be^*}_1/2.
		\end{aligned}
	\end{equation*} 
	Combine the two inequalities and we get $\abs{(\hbe-\be^*)_{S^C}}_1 \leq 3\abs{(\hbe-\be^*)_S}_1$ and this implies $\abs{\hbe-\be^*}_1 \leq 4\abs{(\hbe-\be^*)_S}_1 \leq 4\sqrt{s}\abs{(\hbe-\be^*)_S}_2 \leq  4\sqrt{s}\abs{\hbe-\be^*}_2$.
	
	By the definition of $\hbe$ and the first order condition, we have $\abs{\A\hbe-\b}_\infty\leq \lambda_n$. Combine this with \eqref{cd1} and we have $\abs{\A(\hbe-\be^*)}_\infty \leq 2 \lambda_n$. Together with the condition \eqref{cd2} we have
	\[
	\abs{\hbe-\be^*}_2 \leq c(\hbe-\be^*)^{\rm T} \A(\hbe-\be^*) \leq 2c\lambda_n \abs{\hbe-\be^*}_1 \leq 8c \lambda_n\sqrt{s}\abs{\hbe-\be^*}_2.
	\]
\end{proof}

\subsection{Proof of Some Technical Lemmas}
In this section, we introduce some technical lemmas which will be used in our main proof.\\ 
Let 
\begin{eqnarray}\label{eq:Un}
	U_{n}=\underset{\Abs{\bbeta_{S}-\bbeta_{S}^{\ast}}_{2}\le a_{n}}{\sup}\Big{|}\frac{1}{n}\sum_{i=1}^{n}\left[\bX_{i}\Ind{e_{i}\le \bX_{i,S}^{\rm T}\left(\bbeta_{S}-\bbeta_{S}^{\ast}\right)}-\bX_{i}F\left( \bX_{i,S}^{\rm T}\left(\bbeta_{S}-\bbeta_{S}^{\ast}\right)\right)\right]\nonumber\\
	-\frac{1}{n}\sum_{i=1}^{n}\left[\bX_{i}\Ind{e_{i}\le 0}-\bX_{i}F\left(0\right)\right]\Big{|}_{\infty}.
\end{eqnarray}
\begin{lemma}\label{lem:Un}
	For any $\gamma>0$, there exists a constant $c>0$ such that
	\[
	\pr\left(U_{n}\ge c\sqrt{\frac{sa_{n}\log n}{n}}\right)=O\left(n^{-\gamma}\right).
	\]
\end{lemma}

\begin{proof}[Proof of Lemma \ref{lem:Un}]	
	Let
	\begin{eqnarray*}
		\bC_{nj}(\bbeta)&=&\frac{1}{n}\sum_{k=1}^{n}\left[X_{kj}\Ind{e_{k}\le \bX_{k,S}^{\rm T}\left(\bbeta_{S}-\bbeta_{S}^{\ast}\right)}-X_{kj}F\left( \bX_{k,S}^{\rm T}\left(\bbeta_{S}-\bbeta_{S}^{\ast}\right)\right)\right]\cr
		& &-\frac{1}{n}\sum_{k=1}^{n}\left[X_{kj}\Ind{e_{k}\le 0}-X_{kj}F\left(0\right)\right].
	\end{eqnarray*}
	For notation briefness, we denote $\bbeta^{*}_{S}=(\beta^{*}_{1},\dots,\beta^{*}_{s})^{\rm T}$. For every $i$, we divide the interval $[\beta^{*}_{i}-a_{n},\beta^{*}_{i}+a_{n}]$ into $n^{M}$ small subintervals and each has length $2a_{n}/n^{M}$, where $M$ is a large positive constant. Therefore, there exists a set of points in $\mbR^{p+1}$, $\{\bbeta_{k},1\le k\le q_{n}\}$ with $q_{n}\le n^{Ms}$, such that for any $\bbeta$ in the ball $\abs{\bbeta_{S}-\bbeta^{*}_{S}}_{2} \le a_{n}$, we have $\abs{\bbeta_{S}-\bbeta_{k,S}}_{2}\le 2\sqrt{s}a_{n}/n^{M}$ for some $1\le k\le q_{n}$ and 
	$|\bbeta_{k,S}-\bbeta^{*}_{S}|_{2}\leq a_{n}$. We can see that
	\[
	\Abs{F\left( \bX_{i,S}^{\rm T}\left(\bbeta_{k,S}-\bbeta^{*}_{S}\right)\right)-F\left( \bX_{i,S}^{\rm T}\left(\bbeta_{S}-\bbeta^{*}_{S}\right)\right)}\le C \sqrt{s}a_{n}n^{-M}\Abs{\bX_{i,S}}_{2},
	\]
	and
	\begin{align*}
		&\Abs{\Ind{e_{i}\le \bX_{i,S}^{\rm T}\left(\bbeta_{k,S}-\bbeta^{*}_{S}\right)}-\Ind{e_{i}\le \bX_{i,S}^{\rm T}\left(\bbeta_{S}-\bbeta^{*}_{S}\right)}}\nonumber\\\le&
		\Ind{\bX_{i,S}^{\rm T}\left(\bbeta_{k,S}-\bbeta^{*}_{S}\right)-2\Abs{\bX_{i,S}}_{2}\sqrt{s}a_{n}n^{-M}\le e_{i}\le \bX_{i,S}^{\rm T}\left(\bbeta_{k,S}-\bbeta^{*}_{S}\right)+2\Abs{\bX_{i,S}}_{2}\sqrt{s}a_{n}n^{-M}}\\
		=:&\bG_{i,k}.
	\end{align*}
	Denote the right hand of the above equation by $\bG_{i,k}$ and let $\mbE_{*}(\cdot)$ be the conditional expectation
	given $\{\bX_{i},1\leq i\leq n\}$. Then we have
	\begin{equation*}
		\begin{aligned}
			\mbE_{*}\left(\bG_{i,k}\right)=&F\left(\bX_{i,S}^{\rm T}\left(\bbeta_{k,S}-\bbeta^{*}_{S}\right)+2\Abs{\bX_{i,S}}_{2}\sqrt{s}a_{n}n^{-M}\right)\\&-F\left(\bX_{i,S}^{\rm T}\left(\bbeta_{k,S}-\bbeta^{*}_{S}\right)-2\Abs{\bX_{i,S}}_{2}\sqrt{s}a_{n}n^{-M}\right).
		\end{aligned}
	\end{equation*}

	It is straightforward to conclude that $|\mbE(|X_{ij}|\bG_{i,k})|\le C\sqrt{s}a_{n}n^{-M}\mbE |X_{ij}|\abs{\bX_{i,S}}_{2}\leq Csa_{n}n^{-M}$ and
	$\mbE(X^{2}_{ij}\bG^{2}_{i,k})\leq Csa_{n}n^{-M}$. By the exponential inequality, we can obtain that for any large $\gamma$, there exists a constant $c$ such that
	\begin{eqnarray*}
		\sup_{k}\mbP\Big{(}\frac{1}{n}\Big{|}\sum_{i=1}^{n}(|X_{ij}|\bG_{ik}-\mbE|X_{ij}|\bG_{ik})\Big{|}\geq c\sqrt{\frac{sa_{n}\log n}{n}}\Big{)}\leq
		Cn^{-\gamma s}.
	\end{eqnarray*}
	Note that
	\begin{eqnarray*}
		\sup_{\Abs{\bbeta_{S}-\bbeta^{*}_{S}}_{2}\le a_{n}}\Abs{C_{n,j}\left(\bbeta\right)}-\sup_{k}\Abs{C_{n,j}\left(\bbeta_{k}\right)}
		&\leq& C \sqrt{s}a_{n}n^{-M}\frac{1}{n}\sum_{i=1}^{n}\Abs{\bX_{i,S}}_{2}\\
		&&+\frac{1}{n}\Big{|}\sum_{i=1}^{n}(|X_{ij}|\bG_{ik}-\mbE|X_{ij}|\bG_{ik})\Big{|}\cr
		& &+\frac{1}{n}\Big{|}\sum_{i=1}^{n}\mbE(|X_{ij}|\bG_{ik})\Big{|}.	
	\end{eqnarray*}
	Therefore
	\begin{equation}\label{eqn:Cdiff}
		\sup_{j}\Big{[}\sup_{\Abs{\bbeta_{S}-\bbeta^{*}_{S}}_{2}\le a_{n}}\Abs{C_{n,j}\left(\bbeta\right)}-\sup_{k}\Abs{C_{n,j}\left(\bbeta_{k}\right)}\Big{]}=O_{\pr}\left(\sqrt{\frac{sa_{n}\log n}{n}}\right).
	\end{equation}
	It is enough to show that $\sup_{j}\sup_{k}\abs{C_{n,j}(\bbeta_{k})}$ satisfies the bound in the lemma.  Since the density function of $e_{k}$ is bounded, we have
	\begin{eqnarray*}
		\mbE(C_{n,j}(\bbeta_{k}))^{2}\leq Cn^{-1}\Abs{\bbeta_{k,S}-\bbeta^{*}_{S}}_{2}\leq Cn^{-1}a_{n}.
	\end{eqnarray*}
	By the exponential inequality (Lemma 1 in \cite{Cai-Liu11}) and the fact that $\sqrt{s\log n}=o(\sqrt{na_{n}})$, we have
	\[
	\sup_{j}\sup_{k}\mbP\left(\Abs{\bC_{n,j}\left(\bbeta_{k}\right)}\ge C\sqrt{\frac{sa_{n}\log n}{n}}\right)= O \left(n^{-\gamma s}\right).
	\]
	We complete the proof of the lemma.
\end{proof}
\begin{lemma}\label{prop:f0}
	Assume that (C1)-(C6) hold. Let $\abs{\bhbeta_{0}-\bbeta^{*}}_{2}=O_{\pr}(a_{n})$ and $\pr( supp(\bhbeta_{0})\subseteq S)$ $\rightarrow 1$. Let  $h\geq cs(\log n)/n$ for some $c>0$ and $h=O(a_{n})$. We have 
	\[
	\Abs{\widehat{f}\left(0\right)-f\left(0\right)}=O_{\pr}\left(\sqrt{\frac{s\log n}{nh}}+a_{n}\right).
	\]
\end{lemma}

\begin{proof}[Proof of Lemma \ref{prop:f0}]
	Denote $\widehat{S}$=supp($\bhbeta_{0}$) and let
	\[
	D_{n,h}\left(\bbeta\right)=\frac{1}{nh}\sum_{i=1}^{n}K\left(\frac{Y_{i}-\bX_{i,S}^{\rm T}\bbeta_{S}}{h}\right).
	\]
	We have $\abs{\bbeta^{*}_{S}-\bhbeta_{0,S}}_{2}=O_{\pr}(a_{n})$. To prove the proposition,
	without loss of generality, we can assume that $\abs{\bbeta^{*}_{S}-\bhbeta_{0,S}}_{2}\leq a_{n}$ and $\widehat{S}\subseteq S$.
	It follows that $\widehat{f}(0)=D_{n,h}(\bhbeta_{0})$ and 	
	\[
	\Abs{\widehat{f}\left(0\right)-f\left(0\right)}\le\underset{\abs{\bbeta_{S}-\bbeta^{*}_{S}}_{2}\le a_{n}}{\sup}\Abs{D_{n,h}\left(\bbeta\right)-f\left(0\right)}.
	\]
	Recall the definition of $\{\bbeta_k,1\leq k \leq q_n\}$ in the proof of Lemma \ref{lem:Un}. We have 
	\[
	\Abs{\frac{1}{h}K\left(\frac{Y_{i}-\bX_{i,S}^{\rm T}\bbeta_{S}}{h}\right)-\frac{1}{h}K\left(\frac{Y_{i}-\bX_{i,S}^{\rm T}\bbeta_{k,S}}{h}\right)}\le Ch^{-2}\Abs{\bX_{i,S}^{\rm T}\left(\bbeta_{S}-\bbeta_{k,S}\right)}.
	\]
	This yields that
	\[
	\underset{\abs{\bbeta_{S}-\bbeta^{*}_{S}}_{2}\le a_{n}}{\sup}\Abs{D_{n,h}\left(\bbeta\right)-f\left(0\right)}-\underset{1\leq k\leq q_{n}}{\sup}\Abs{D_{n,h}\left(\bbeta_{k}\right)-f\left(0\right)}\le\frac{C\sqrt{s}a_{n}}{n^{M+1}h^{2}}\sum_{i=1}^{n}\Abs{\bX_{i,S}}_{2}.
	\]
	Since $\max_{i,j}\mbE\abs{X_{i,j}}^{2}<\infty$ (due to the sub-Gaussian condition (C4)), for any $\gamma>0$, by letting $M$ large enough, we have 
	\begin{equation}\label{eqn:Ddiff}
		\underset{\abs{\bbeta_{S}-\bbeta^{*}_{S}}_{2}\le a_{n}}{\sup}\Abs{D_{n,h}\left(\bbeta\right)-f\left(0\right)}-\underset{1\leq k\leq n^{Ms}}{\sup}\Abs{D_{n,h}\left(\bbeta_{k}\right)-f\left(0\right)}=O_{\pr}\left(n^{-\gamma}\right).
	\end{equation}
	It is enough to show that $\sup_{k}\abs{D_{n,h}(\bbeta_{k})-\mbE D_{n,h}(\bbeta_{k})}$ and $\sup_{k}\abs{\mbE D_{n,h}(\bbeta_{k})-f(0)}$  satisfy the bound in the proposition. Let $\mbE_{\ast}(\cdot)$ denote the conditional expectation given $\{\bX_{k}\}$.
	We have
	\begin{align*}
		\mbE_{\ast}\left\{\frac{1}{h}K\left(\frac{e_{i}-\bX_{i,S}^{\rm T}\left(\bbeta_{S}-\bbeta^{*}_{S}\right)}{h}\right)\right\}=&\int_{-\infty}^{\infty}K\left(x\right)f\left\{hx+\bX_{i,S}^{\rm T}\left(\bbeta_{S}-\bbeta^{*}_{S}\right)\right\}dx\\
		=&f\left(0\right)+O\left(h+\Abs{\bX_{i,S}^{\rm T}\left(\bbeta_{S}-\bbeta^{*}_{S}\right)}\right).
	\end{align*}
	Since $\sup_{|\balpha|_{2}=1}\mbE |\balpha^{\rm T}\bX|\leq C$, we have
	\begin{align*}
		\Abs{\mbE D_{n,h}\left(\bbeta_{k}\right)-f\left(0\right)}\le C\left(h+\Abs{\bbeta_{k,S}-\bbeta^{*}_{S}}_{2}\right)=O(h+a_{n}).
	\end{align*}
	It remains to bound  $\sup_{k}\abs{D_{n,h}(\bbeta_{k})-\mbE D_{n,h}(\bbeta_{k})}$. Put	\[
	\xi_{i,k}=K\left(\frac{e_{i}-\bX_{i,S}^{\rm T}\left(\bbeta_{k,S}-\bbeta^{*}_{S}\right)}{h}\right).
	\]
	We have
	\begin{eqnarray*}
		\mbE_{\ast} \xi^{2}_{i,k}
		=h\int_{-\infty}^{\infty}\left\{K\left(x\right)\right\}^{2}f\left\{hx+\bX_{i,S}^{\rm T}\left(\bbeta_{k,S}-\bbeta^{*}_{S}\right)\right\}dx\le Ch.
	\end{eqnarray*}
	Since $K(x)$ is bounded, we have,  
	by the exponential inequality (Lemma 1 in \cite{Cai-Liu11}) and the fact that $s\log n=O(nh)$, for any $\gamma>0$, there exists a constant $C>0$ such that
	\[
	\sup_{k}\mbP\left(\left |\sum_{i=1}^{n}(\xi_{i,k}-\mbE \xi_{i,k}) \right |\ge C\sqrt{nhs\log n}\right)= O \left(n^{-\gamma s}\right).
	\]
	By letting $\gamma>M$, we can obtain that
	\[
	\Abs{\sup_{k}\abs{D_{n,h}(\bbeta_{k})-\mbE D_{n,h}(\bbeta_{k})}}=O_{\pr}\left(\sqrt{\frac{s\log n}{nh}}\right).
	\]
	This completes the proof.
\end{proof}

\begin{lemma} \label{lem:str}
	We have
	\begin{eqnarray*}
		\max_{1\leq j\leq p}\left\|n^{-1}\sum_{k=1}^{n}|X_{kj}|\bX_{k,S}\bX_{k,S}^{\rm T}\right\|_{\mathrm{op}}=O_{\pr}(1).
	\end{eqnarray*}
\end{lemma}

\begin{proof}[Proof of Lemma \ref{lem:str} ] 
	
	For a unit ball $B$ in $\R^{s}$, we have the fact that there exist $q_{s}$ balls with centers $\x_{1},\ldots,\x_{q_{s}}$ and radius $z$ (i.e., $B_{i}=\{\x\in \R^{s}:|\x-\x_{i}|\leq z\}$, $1\leq i\leq q_{s}$) such that $B\subseteq \cup_{i=1}^{q_{s}}B_{i}$ and $q_{s}$ satisfies $q_{s}\leq (1+2/z)^{s}$. 
	So for any $|\x|_{2}=1$ in the unit sphere, there exists some $\x_{i}$ such that $|\x-\x_{i}|_{2}\leq z$ and so this $\x_{i}$ satisfies $1-z\leq |\x_{i}|_{2}\leq 1+z$.
	Therefore,  there exists a subset $K\subset\{1,2,\ldots,q_{s}\}$ such that $\{\x:|\x|_{2}=1\}\subseteq \cup_{i\in K}B_{i}$ and $1-z\leq |\x_{i}|_{2}\leq 1+z$ for $i\in K$.
	We have $d_{s}:=|K|\leq q_{s}\leq (1+2/z)^{s}$.
	
	For any $s\times s$ symmetric matrix $\A$, we have 
	\begin{eqnarray*}
		|\x^{\mathrm{T}}\A\x|-|\y^{\mathrm{T}}\A\y|\leq |(\x-\y)^{\mathrm{T}}\A(\x+\y)|.
	\end{eqnarray*}
	So $\norm{\A}_{\mathrm{op}}=\sup_{|\x|_{2}=1}|\x^{\mathrm{T}}\A\x|\leq \max_{i\in K}|\x^{\mathrm{T}}_{i}\A\x_{i}|+ z(2+z)\norm{\A}_{\mathrm{op}}$. Now take $z=1/4$, we have 
	$\norm{\A}_{\mathrm{op}}\leq 3\max_{i\in K}|\x^{\mathrm{T}}_{i}\A\x_{i}|$ and $d_{s}\leq 9^{s}$. It is enough to prove that
	\begin{eqnarray*}
		\max_{1\leq j\leq p}\max_{i\in K}\frac{1}{n}\sum_{k=1}^{n}|X_{kj}|(\x^{\mathrm{T}}_{i}\bX_{k,S})^{2}=O_{\pr}(1).
	\end{eqnarray*}
	Define $\widehat{X}_{kj}=X_{kj}\ind{|X_{kj}|\leq \log n}$.
	By the sub-Gaussian condition on $\X$, it is enough to show that
	\begin{eqnarray*}
		\max_{1\leq j\leq p}\max_{i\in K}\frac{1}{n}\sum_{k=1}^{n}|\widehat{X}_{kj}|(\x^{\mathrm{T}}_{i}\bX_{k,S})^{2}=O_{\pr}(1).
	\end{eqnarray*}
	Set 
	\begin{eqnarray*}
		Y_{kij}=|\widehat{X}_{kj}|(\x^{\mathrm{T}}_{i}\bX_{k,S})^{2} \ind{|\widehat{X}_{kj}|(\x^{\mathrm{T}}_{i}\bX_{k,S})^{2}\leq (s+1)(\log n)^{3}}.
	\end{eqnarray*}
	Note that
	\begin{eqnarray*}
		np9^{s}\max_{k,j}\max_{i\in K}\pr\left(|\widehat{X}_{kj}|(\x^{\mathrm{T}}_{i}\bX_{k,S})^{2}\geq (s+1)(\log n)^{3}\right)=o(1).
	\end{eqnarray*}
	It suffices to prove that $\max_{1\leq j\leq p}\max_{i\in K}\frac{1}{n}\sum_{k=1}^{n}Y_{kij}=O_{p}(1)$. It is easy to see that
	$\mbE Y_{kij}\leq \mbE |X_{kj}|(\x^{\mathrm{T}}_{i}\bX_{k,S})^{2}\leq C(\mbE X^{2}_{kj})^{1/2}\sup_{|\x|_{2}=1}(\mbE (\x^{\mathrm{T}}\bX_{k,S})^{4})^{1/2}=O(1)$ and similarly, 
	$\mbE Y^{2}_{kij}=O(1)$, uniformly in $k,i,j$.
	By Bernstein's inequality,
	\begin{eqnarray*}
		\pr\left(\Big{|}\frac{1}{n}\sum_{k=1}^{n}(Y_{kij}-\mbE Y_{kij})\Big{|}\geq 1\right)\leq e^{-c_{1}n}+e^{-c_{2}\frac{n}{(s+1)(\log n)^{3}}},
	\end{eqnarray*}
	for some positive constants $c_{1}$ and $c_{2}$ uniformly in $i,j$. Since $s=O(m^{r})$ for some $0<r<1/3$, we have 
	\begin{eqnarray*}
		np9^{s}\left(e^{-c_{1}n}+e^{-c_{2}\frac{n}{(s+1)(\log n)^{3}}}\right)=o(1).
	\end{eqnarray*}
	This proves $\max_{1\leq j\leq p}\max_{i\in K}\frac{1}{n}\sum_{k=1}^{n}Y_{kij}=O_{p}(1)$.
\end{proof}

\subsection{Proof of Theorem \ref{thm:betainf} and Theorem \ref{thm:betainft}}
We first state a proposition for the proof of our main theorems.

\begin{proposition}\label{prop:Bnbeta0}
	Assume that (C1)-(C6) hold. Let $\abs{\bhbeta_{0}-\bbeta^{*}}_{2}=O_{\pr}(a_{n})$ and $h\asymp a_{n}$. We have 
	\[
	\Abs{\z_{n}-\bhSigma\bbeta^{*}}_{\infty}=O_{\pr}\left(\sqrt{\frac{\log n}{n}}+a_{n}^{2}\right).
	\]
\end{proposition}

\begin{proof}[Proof of Proposition \ref{prop:Bnbeta0}]
	Recall the definition of $U_n$ in \eqref{eq:Un}.
	For the initial estimator, we have $\bhbeta_{0,S^{c}}=0$ with high probability. Due to the fact that $\bbeta^{*}_{S^{c}}=\bm{0}$ and $\bbeta_{0,S^{c}}=\bm{0}$, by $\abs{\bbeta^{*}-\bhbeta_{0}}_{2}=O_{\pr}(a_{n})$, we have 
	\begin{eqnarray*}
		&&\Abs{\z_{n}-\bhSigma\bbeta^{*}}_{\infty}\cr &=&\Abs{-\frac{\widehat{f}^{-1}\left(0\right)}{n}\sum_{k=1}^{n}\bX_{k}\left(\Ind{Y_{k}\le \bX_{k}^{\rm T}\bhbeta_{0}}-\tau\right)+\bhSigma\left(\bhbeta_{0}-\bbeta^{*}\right)}_{\infty}\\
		&\le& \Abs{\frac{\widehat{f}^{-1}\left(0\right)}{n}\sum_{k=1}^{n}\bX_{k}\left\{F\left( \bX_{k,S}^{\rm T}\left(\bbeta_{S}-\bhbeta_{0,S}\right)\right)-F\left(0\right)\right\}+\frac{1}{n}\sum_{k=1}^{n}\bX_{k}\bX_{k,S}^{\rm T}\left(\bhbeta_{0,S}-\bbeta^{*}_{S}\right)}_{\infty}\\
		&&+\Abs{\widehat{f}^{-1}(0)}\Abs{\frac{1}{n}\sum_{k=1}^{n}\left[\bX_{k}\Ind{e_{k}\le 0}-\bX_{k}F\left(0\right)\right]}_{\infty}+\Abs{\widehat{f}^{-1}\left(0\right)}U_{n}.
	\end{eqnarray*}
	For the last term, by Lemma \ref{lem:Un}, we have $\abs{\widehat{f}^{-1}\left(0\right)}U_{n}=O_{\pr}(\sqrt{s a_{n}(\log n)/n})$.
	For the second term of the right hand side, we have 
	\begin{eqnarray*}
		\Abs{\widehat{f}^{-1}(0)}\Abs{\frac{1}{n}\sum_{k=1}^{n}\left[\bX_{k}\Ind{e_{k}\le 0}-\bX_{k}F\left(0\right)\right]}_{\infty}=O_{\pr}\Big{(}\sqrt{\frac{\log p}{n}}\Big{)}.
	\end{eqnarray*}
	Denote the first term of the right hand side to be $\bm{H}$. For the first component of $\bm{H}$, by second order Taylor expansion, under (C1) we have
	\begin{align*}
		&\frac{\widehat{f}^{-1}\left(0\right)}{n}\sum_{k=1}^{n}X_{kj}\left\{F\left( \bX_{k,S}^{\rm T}\left(\bbeta_{S}^{*}-\bhbeta_{0,S}\right)\right)-F\left(0\right)\right\}\\
		=& \frac{\widehat{f}^{-1}\left(0\right)f\left(0\right)}{n}\sum_{k=1}^{n}X_{kj}\bX_{k,S}^{\rm T}\left(\bbeta^{*}_{S}-\bhbeta_{0,S}\right)+\frac{C\widehat{f}^{-1}\left(0\right)}{n}\sum_{k=1}^{n}|X_{kj}|\left\{\bX_{k,S}^{\rm T}\left(\bbeta^{*}_{S}-\bhbeta_{0,S}\right)\right\}^2.
	\end{align*}
	It is standard to show that 
	\begin{eqnarray*}
		\pr\left(|\bhSigma-\bSigma|_{\infty}\leq C\sqrt{\frac{\log n}{n}}\right)\rightarrow 1.
	\end{eqnarray*}
	Since $\Lambda_{\text{max}}(\bSigma)\leq c_{0}$,  we have
	\begin{eqnarray*}
		\Abs{\frac{1}{n}\sum_{k=1}^{n}\bX_{k}\bX_{k,S}^{\rm T}\left(\bbeta^{*}_{S}-\bhbeta_{0,S}\right)}_{\infty}&\leq& O_{\pr}\left(\sqrt{\frac{s\log n}{n}}a_{n}\right)
		+\Abs{\bSigma\left(\bbeta^{*}_{S}-\bhbeta_{0,S}\right)}_{\infty}\cr
		&=&O_{\pr}(a_{n}).
	\end{eqnarray*}
	
	Denote $(1,|X_{k1}|,\ldots,|X_{kp}|)^{\mathrm{T}}$ by $|\bX_{k}|$.
	Then by Lemma \ref{prop:f0} and \ref{lem:str}, we have 
	\begin{align*}
		\Abs{\bm{H}}_{\infty}\le&\Abs{\widehat{f}^{-1}\left(0\right)f\left(0\right)-1}\Abs{\frac{1}{n}\sum_{k=1}^{n}\bX_{k}\bX_{k,S}^{\rm T}\left(\bbeta^{*}_{S}-\bhbeta_{0,S}\right)}_{\infty}\\
		&+C\widehat{f}^{-1}\left(0\right)\Abs{\frac{1}{n}\sum_{k=1}^{n}|\bX_{k}|\left\{\bX_{k,S}^{\rm T}\left(\bbeta^{*}_{S}-\bhbeta_{0,S}\right)\right\}^2}_{\infty}\\
		=&O_{\pr}\left(\left(\sqrt{\frac{s\log n}{nh}}+a_{n}\right)a_{n}\right)+O_{\pr}(a^{2}_{n}).
	\end{align*}
	So we can easily have 
	\[
	\Abs{\z_{n}-\bhSigma\bbeta^{*}}_{\infty}=O_{\pr}\left(\sqrt{\frac{\log p}{n}}+\sqrt{\frac{sa_{n}\log n}{n}}+a_{n}\sqrt{\frac{s\log n}{nh}}+a_{n}^{2}\right).
	\]
	Since $h\asymp a_{n}$ and $sa_{n}=o(1)$, we prove the proposition.
\end{proof}

\begin{proof}[Proof of Theorem \ref{thm:betainf} and Theorem \ref{thm:betainft}]
	
	First, we show the results for Theorem \ref{thm:betainf}. Define $\btbeta$ to be the solution of the following optimization problem:
	\[
	\btbeta=\argmin_{\btheta\in\mbR^{p+1},\btheta_{S^{c}}=0}\frac{1}{2}\btheta^{\rm T}\bhSigma_{1}\btheta-\btheta^{\rm T}\left\{\z_{n}+\left(\bhSigma_{1}-\bhSigma\right)\bhbeta_{0}\right\}+\lambda_{n}\Abs{\btheta}_{1},
	\]
	where $\btheta_{S^{c}}$ denotes the subset vector with the coordinates of $\btheta$ in $S^{c}$. Then there exist sub-gradients $\btZ$ with $\abs{\btZ}_{\infty}\le 1$ such that
	\begin{equation}\label{eqn:subg}
		\bhSigma_{1,S\times S}\btbeta_{S}-\left\{\z_{n}+\left(\bhSigma_{1}-\bhSigma\right)\bhbeta_{0}\right\}_{S}+\lambda_{n}\btZ_{S}=0.
	\end{equation}
	It is enough to show that there exist sub-gradients $\bZ$ that satisfy
	\begin{equation}\label{eqn:grad}
		\bhSigma_{1}\bhbeta-\left\{\z_{n}+\left(\bhSigma_{1}-\bhSigma\right)\bhbeta_{0}\right\}+\lambda_{n}\bZ=0,
	\end{equation}
	$\abs{\bZ_{S}}_{\infty}\leq 1$ and $\abs{\bZ_{S^{c}}}_{\infty}<1$, i.e., $\abs{Z_{i}}$ are strictly less than one for $i\in S^{c}$. To construct such $\bZ$, we let $\bZ_{S}=\btZ_{S}$ and 
	\begin{eqnarray*}
		\bZ_{S^{c}}=&-\lambda_{n}^{-1}\left\{\left(\bhSigma_{1}\btbeta\right)_{S^{c}}-\left\{\z_{n}+\left(\bhSigma_{1}-\bhSigma\right)\bhbeta_{0}\right\}_{S^{c}}\right\}.
	\end{eqnarray*}
	\begin{lemma}\label{lem:strictless}
		Under the conditions in Theorem \ref{thm:betainf}, we have, with probability tending to one,
		\[
		\Abs{Z_{i}}\leq v
		\]
		uniformly for $i\in S^{c}$, for some $0<v<1$.
	\end{lemma}	
	\begin{proof}[Proof of Lemma \ref{lem:strictless}]
		Recall that
		\begin{equation}\label{eqn:ZS}
			\bhSigma_{1,S\times S}\btbeta_{S}-\left\{\z_{n}+\left(\bhSigma_{1}-\bhSigma\right)\bhbeta_{0}\right\}_{S}=-\lambda_{n}\btZ_{S}.
		\end{equation}
		Write \eqref{eqn:ZS} as 
		\begin{equation*}
			\begin{aligned}
				-\lambda_{n}\btZ_{S}=&\bSigma_{S\times S}\left(\btbeta_{S}-\bbeta^{*}_{S}\right)+\left(\bhSigma_{1,S\times S}-\bSigma_{S\times S}\right)\left(\btbeta_{S}-\bbeta^{*}_{S}\right)+\bhSigma_{1,S\times S}\bbeta^{*}_{S}\\&-\left\{\z_{n}+\left(\bhSigma_{1}-\bhSigma\right)\bhbeta_{0}\right\}_{S}.
			\end{aligned}
		\end{equation*}

		This implies that
		\begin{eqnarray*}
			\btbeta_{S}-\bbeta^{*}_{S}&=&\bSigma_{S\times S}^{-1}\Big{\{}-\lambda_{n}\btZ_{S}-\left(\bhSigma_{1,S\times S}-\bSigma_{S\times S}\right)\left(\btbeta_{S}-\bbeta^{*}_{S}\right)\cr
			& &\quad-\bhSigma_{1,S\times S}\bbeta_{S}+\left\{\z_{n}+\left(\bhSigma_{1}-\bhSigma\right)\bhbeta_{0}\right\}_{S}\Big{\}}\\
			&=&\bSigma_{S\times S}^{-1}\Big{\{}-\lambda_{n}\btZ_{S}-\left(\bhSigma_{1,S\times S}-\bSigma_{S\times S}\right)\left(\btbeta_{S}-\bbeta^{*}_{S}\right)\cr
			&&\quad-\left(\bhSigma_{1,S\times S}-\bhSigma_{S\times S}\right)\left(\bbeta^{*}_{S}-\bhbeta_{0,S}\right)+\left(\z_{n}-\bhSigma\bbeta^{*}\right)_{S}\Big{\}}.
		\end{eqnarray*}
		By (\ref{ads}), we have with probability tending to one,
		\begin{eqnarray*}
			\Abs{\btbeta_{S}-\bbeta^{*}_{S}}_{2}&\le& C\sqrt{s}\lambda_{n}+ C\sqrt{\frac{s\log n}{m}}\Abs{\btbeta_{S}-\bbeta^{*}_{S}}_{2}\cr
			&&+C\sqrt{s}\left(\sqrt{\frac{\log n}{m}}+\sqrt{\frac{\log n}{n}}\right)\Abs{\bbeta^{*}_{S}-\bhbeta_{0,S}}_{2}+C\sqrt{s}\Abs{\z_{n}-\bhSigma\bbeta^{*}}_{\infty}.
		\end{eqnarray*}
		By the choice of  $\lambda_{n}$, Proposition \ref{prop:Bnbeta0} and $a_{n}=O(\sqrt{s(\log n)/m})$,  
		\begin{equation}\label{eqn:btbeta}
			\Abs{\btbeta_{S}-\bbeta^{*}_{S}}_{2}\le C\sqrt{s}\lambda_{n},
		\end{equation}
		with probability tending to one.
		
		Due to the definition of $\bZ_{S^{c}}$, we have that
		\begin{align}\label{zb}
			&\nonumber\bZ_{S^{c}}\\\nonumber=&-\lambda_{n}^{-1}\left\{\left(\bhSigma_{1}\btbeta\right)_{S^{c}}-\left\{\z_{n}+\left(\bhSigma_{1}-\bhSigma\right)\bhbeta_{0}\right\}_{S^{c}}\right\}\\
			\nonumber=&-\lambda_{n}^{-1}\bhSigma_{1,S^{c}\times S}\bhSigma_{1,S\times S}^{-1}\left\{\z_{n}+\left(\bhSigma_{1}-\bhSigma\right)\bhbeta_{0}\right\}_{S}+\bhSigma_{1,S^{c}\times S}\bhSigma_{1,S\times S}^{-1}\btZ_{S}\\
			\nonumber&+\lambda_{n}^{-1}\left\{\z_{n}+\left(\bhSigma_{1}-\bhSigma\right)\bhbeta_{0}\right\}_{S^{c}}\\
			\nonumber=&-\lambda_{n}^{-1}\bhSigma_{1,S^{c}\times S}\bhSigma_{1,S\times S}^{-1}\left[\left\{\z_{n}-\bhSigma\bbeta^{*}\right\}_{S}+\left(\bhSigma_{S\times \{1,\ldots,p+1\}}-\bhSigma_{1,S\times \{1,\ldots,p+1\}}\right)\left(\bbeta^{*}-\bhbeta_{0}\right)\right]\\
			\nonumber&+\bhSigma_{1,S^{c}\times S}\bhSigma_{1,S\times S}^{-1}\btZ_{S}
			+\lambda_{n}^{-1}\left\{\z_{n}-\bhSigma\bbeta^{*}\right\}_{S^{c}}\\
			\nonumber&+\lambda_{n}^{-1}\left(\bhSigma_{S^{c}\times \{1,\ldots,p+1\}}-\bhSigma_{1,S^{c}\times \{1,\ldots,p+1\}}\right)\left(\bbeta^{*}-\bhbeta_{0}\right).\\
		\end{align}
		Note that 
		\begin{align*}
			&\bhSigma_{1,S^{c}\times S}\bhSigma_{1,S\times S}^{-1}\\=&\left(\bhSigma_{1,S^{c}\times S}-\bSigma_{S^{c}\times S}\right)\left(\bhSigma_{1,S\times S}^{-1}-\bSigma_{S\times S}^{-1}\right)+\bSigma_{S^{c}\times S}\left(\bhSigma_{1,S\times S}^{-1}-\bSigma_{S\times S}^{-1}\right)\\
			&+\left(\bhSigma_{1,S^{c}\times S}-\bSigma_{S^{c}\times S}\right)\bSigma_{S\times S}^{-1}+\bSigma_{S^{c}\times S}\bSigma_{S\times S}^{-1}.
		\end{align*}
		By the proof of Lemma \ref{lem:str}, we can easily get
		\begin{eqnarray*}
			\Big{\|}\bhSigma_{1,S\times S}-\bSigma_{S\times S}\Big{\|}_{\mathrm{op}}=O_{\pr}\Big{(}\sqrt{\frac{s+\log n}{m}}\Big{)}.
		\end{eqnarray*}
		This yields that
		\begin{eqnarray*}
			\Big{\|}\bhSigma^{-1}_{1,S\times S}-\bSigma^{-1}_{S\times S}\Big{\|}_{\mathrm{op}}=O_{\pr}\Big{(}\sqrt{\frac{s+\log n}{m}}\Big{)}.
		\end{eqnarray*}
		Then
		\begin{align*}
			&\Norm{\left(\bhSigma_{1,S^{c}\times S}-\bSigma_{S^{c}\times S}\right)\left(\bhSigma_{1,S\times S}^{-1}-\bSigma_{S\times S}^{-1}\right)}_{\infty}
			\\\leq& s^{3/2}\Abs{\bhSigma_{1,S^{c}\times S}-\bSigma_{S^{c}\times S}}_{\infty}\Big{\|}\bhSigma^{-1}_{1,S\times S}-\bSigma^{-1}_{S\times S}\Big{\|}_{\mathrm{op}}\\
			=&O_{\pr}\left(s^{2}(\log n)/m\right).
		\end{align*}
		Similarly,
		\begin{align*}
			\Norm{\bSigma_{S^{c}\times S}\left(\bhSigma_{1,S\times S}^{-1}-\bSigma_{S\times S}^{-1}\right)}_{\infty}\leq& s\|\bSigma\|_{\mathrm{op}}\Big{\|}\bhSigma^{-1}_{1,S\times S}-\bSigma^{-1}_{S\times S}\Big{\|}_{\mathrm{op}}\\ =&O_{\pr}\left(s\sqrt{\frac{s+\log n}{m}}\right),
		\end{align*}
		and
		\begin{align*}
			\Norm{\left(\bhSigma_{1,S^{c}\times S}-\bSigma_{S^{c}\times S}\right)\bSigma_{S\times S}^{-1}}_{\infty}\leq& s^{3/2}\Abs{\bhSigma_{1,S^{c}\times S}-\bSigma_{S^{c}\times S}}_{\infty}\Big{\|}\bSigma^{-1}_{S\times S}\Big{\|}_{\mathrm{op}}\\=&O_{\pr}\left(\sqrt{\frac{s^{3}\log n}{m}}\right).
		\end{align*}
		So we have $\norm{\bhSigma_{1,S^{c}\times S}\bhSigma_{1,S\times S}^{-1}}_{\infty}\le o_{\pr}(1)+\norm{\bSigma_{S^{c}\times S}\bSigma_{S\times S}^{-1}}_{\infty}$. Since $C_{0}$ in $\lambda_{n}$ is sufficiently large, we can see that $\lambda_{n}^{-1}\abs{\z_{n}-\bhSigma\bbeta^{*}}_{\infty}$ is small enough.
		
		Since $\norm{\bSigma_{S^{c}\times S}\bSigma_{S\times S}^{-1}}_{\infty}\le 1-\alpha$ and $\abs{\btZ_{S}}_{\infty}\le 1$, we have $\abs{\bhSigma_{1,S^{c}\times S}\bhSigma_{1,S\times S}^{-1}\btZ_{S}}_{\infty}\le 1-\alpha/2$ with probability tending to one. Note that 
		$\pr(\text{supp}(\bhbeta_{0})\subseteq S)\rightarrow 1$, we have	
		\begin{align*}
			&\lambda_{n}^{-1}\Abs{\left(\bhSigma_{1,S\times \{1,\ldots,p+1\}}-\bhSigma_{S\times \{1,\ldots,p+1\}}\right)\left(\bbeta^{*}-\bhbeta_{0}\right)}_{\infty}\\=& O_{\pr}(1) \lambda_{n}^{-1}\sqrt{s(\log n)/m}\Abs{\bbeta^{*}-\bhbeta_{0}}_{2}\\
			=& O_{\pr}\Big{(}\lambda_{n}^{-1}a_{n}\sqrt{s(\log n)/m}\Big{)}\\
			=&O_{\pr}(1/C_{0}),
		\end{align*}
		and
		\begin{align*}
			&\Abs{\left(\bhSigma_{1,S^{c}\times \{1,\ldots,p+1\}}-\bhSigma_{S^{c}\times \{1,\ldots,p+1\}}\right)\left(\bbeta^{*}-\bhbeta_{0}\right)}_{\infty}\\=& 
			O_{\pr}(1) \lambda_{n}^{-1}\sqrt{s(\log n)/m}\Abs{\bbeta^{*}-\bhbeta_{0}}_{2}\\
			=&O_{\pr}(1/C_{0}).
		\end{align*}
		The above arguments, together with (\ref{zb}), imply uniformly for $j\in S^{c}$ and some $v<1$,
		\[
		\Abs{Z_{j}}\le v<1.
		\]
	\end{proof}
	
	By Lemma \ref{lem:strictless}, uniformly for $i\in S^{c}$ and some $v<1$,
	\[
	\Abs{Z_{i}}\le v<1
	\]
	with probability tending to one. By this primal-dual witness construction, we have $\bhbeta=\btbeta$ with probability tending to one.
	Thus 
	\begin{eqnarray*}
		\pr\left(\abs{\bhbeta-\bbeta^{*}}_{1}\le\sqrt{s}\abs{\bhbeta-\bbeta^{*}}_{2}\right)\rightarrow 1.
	\end{eqnarray*}
	It is easy to see that 
	\[
	\Abs{\bhSigma_{1}\bhbeta-\z_{n}-\left(\bhSigma_{1}-\bhSigma\right)\bhbeta_{0}}_{\infty}\le \lambda_{n},
	\]
	due to equation \eqref{eqn:subg} and $\abs{\bZ}_{\infty}\le1$. 
	It is standard to show that for some $C>0$,
	\begin{eqnarray}\label{ads}
		\pr\left(|\bhSigma_{1}-\bSigma|_{\infty}\leq C\sqrt{\frac{\log n}{m}}\right)\rightarrow 1.
	\end{eqnarray}
	Note that $\pr(\text{supp}(\bhbeta_{0})\subseteq S)\rightarrow 1$. By Proposition \ref{prop:Bnbeta0},
	\begin{eqnarray}\label{easytoshow}
		\abs{\bhSigma_{1}\bbeta^{*}-\z_{n}-(\bhSigma_{1}-\bhSigma)\bhbeta_{0}}_{\infty}&\le& \abs{\z_{n}-\bhSigma\bbeta^{*}}_{\infty}
		+\abs{(\bhSigma_{1}-\bhSigma)(\bhbeta_{0}-\bbeta^{*})}_{\infty}\cr
		&=&O_{\pr}\left(\sqrt{\frac{\log n}{n}}+a_{n}^{2}+\sqrt{\frac{s\log n}{m}}a_{n}\right)\cr
		&=&\frac{O_{\pr}(1)}{C_{0}}\lambda_{n}.
	\end{eqnarray}
	Therefore, by letting $C_{0}$ in $\lambda_n$ being sufficiently large,  we have $\abs{\bhSigma_{1}(\bbeta^{*}-\bhbeta)}_{\infty}\le 2\lambda_{n}$ with probability tending to one. By the following condition 
	\begin{equation}\label{eqn:cond2}
		\underset{\delta:\abs{\delta}_{1}\le c_{1}\sqrt{s}\abs{\delta}_{2}}{\min}\frac{\delta^{\rm T}\bhSigma_{1}\delta}{\abs{\delta}_{2}^{2}}\ge c_{2},\qquad c_1,c_2>0,
	\end{equation}
	we can further have
	\[
	\Abs{\bhbeta-\bbeta^{*}}_{2}^{2}\le C\left(\bhbeta-\bbeta^{*}\right)^{\rm T}\bhSigma_1\left(\bhbeta-\bbeta^{*}\right)\le C\lambda_{n}\Abs{\bhbeta-\bbeta^{*}}_{1}\le C\lambda_{n}\sqrt{s}\Abs{\bhbeta-\bbeta^{*}}_{2}.
	\]
	This proves that $\abs{\bhbeta-\bbeta^{*}}_{2}\le C\lambda_{n}\sqrt{s}$.
	
	To prove Theorem \ref{thm:betainf}, it is enough to show that $\bhSigma_{1}$ satisfies condition \eqref{eqn:cond2}. We have, with probability tending to one,
	\begin{align*}
		\delta^{\rm T}\bhSigma_{1}\delta\ge&\Abs{\delta}_{2}^{2}\lambda_{\min}\left(\bSigma\right)-\Abs{\delta}_{1}^{2}\Abs{\bhSigma_{1}-\bSigma}_{\infty}\\
		\ge&\Abs{\delta}_{2}^{2}\lambda_{\min}\left(\bSigma\right)-\Abs{\delta}_{2}^{2}s\Abs{\bhSigma_{1}-\bSigma}_{\infty}\\
		\ge&c\Abs{\delta}_{2}^{2},
	\end{align*}
	for some $c>0$ as $s=o((m/\log n)^{1/2})$. This completes the proof of Theorem \ref{thm:betainf}.
	
	For $t=1$, note that we assume $\abs{\bhbeta_{0}-\bbeta^{*}}_{2}=O_{\pr}(\sqrt{s(\log n)/m})$. Then let $a_{n}=\sqrt{s(\log n)/m}$ in Theorem \ref{thm:betainf} and it is easy to see Theorem \ref{thm:betainft} holds for $t=1$. Now suppose Theorem \ref{thm:betainft} holds for $t=k-1$ with some $k\ge 2$. Then for $t=k$ with initial estimator being $\bhbeta^{(k-1)}$, we have $a_{n,k-1}=\sqrt{\frac{s \log n}{n}}+s^{(2k-1)/2}\left(\frac{\log n}{m}\right)^{k/2}$.
	Hence by Theorem \ref{thm:betainf} again and the condition on $s$,
	\begin{eqnarray*}
		\Abs{\bhbeta^{(k)}-\bbeta^{*}}_{2}&=&O_{\pr}\left(\sqrt{\frac{s\log n}{n}}+a_{n,k-1}\sqrt{\frac{s^{2}\log n}{m}}\right)\cr
		&=&O_{\pr}\left(\sqrt{\frac{s \log n}{n}}+s^{(2k+1)/2}\left(\frac{\log n}{m}\right)^{(k+1)/2}\right).
	\end{eqnarray*}
	This implies that Theorem \ref{thm:betainft} holds for $t=k$. Then it completes the proof of Theorem \ref{thm:betainft}.
\end{proof}

\subsection{Proof of Theorem \ref{thm:support} and Theorem \ref{thm:supportt}}
\begin{proof}[Proof of Theorem \ref{thm:support} and Theorem \ref{thm:supportt}]
	Theorem \ref{thm:support} (i) and \ref{thm:supportt} (i) follow directly from the proof of Theorem \ref{thm:betainf}.
	As for Theorem \ref{thm:support} (ii), note that $\pr(\bhbeta=\btbeta)\rightarrow 1$. Recall \begin{eqnarray*}
		\btbeta_{S}-\bbeta^{*}_{S}
		&=&\bSigma_{S\times S}^{-1}\Big{\{}-\lambda_{n}\btZ_{S}-\left(\bhSigma_{1,S\times S}-\bSigma_{S\times S}\right)\left(\btbeta_{S}-\bbeta^{*}_{S}\right)\cr
		&&\quad-\left(\bhSigma_{1,S\times S}-\bhSigma_{S\times S}\right)\left(\bbeta^{*}_{S}-\bhbeta_{0,S}\right)+\left(\z_{n}-\bhSigma\bbeta^{*}\right)_{S}\Big{\}}.
	\end{eqnarray*}	
	By Equation \eqref{eqn:btbeta}, we obtain that, with probability tending to one,
	\begin{eqnarray*}
		\Abs{\bSigma_{S\times S}^{-1}\left(\bhSigma_{1,S\times S}-\bSigma_{S\times S}\right)\left(\btbeta_{S}-\bbeta^{*}_{S}\right)}_{\infty}&\le& C\|\bSigma_{S\times S}^{-1}\|_{\infty}\|\bhSigma_{1,S\times S}-\bSigma_{S\times S}\|_{\mathrm{op}}\Abs{\btbeta_{S}-\bbeta^{*}_{S}}_{2}\cr
		&\leq &C\|\bSigma_{S\times S}^{-1}\|_{\infty}\sqrt{s(s+\log n)/m}\Abs{\btbeta_{S}-\bbeta^{*}_{S}}_{\infty},
	\end{eqnarray*}
	\[
	\Abs{\bSigma_{S\times S}^{-1}\left(\bhSigma_{1,S\times S}-\bhSigma_{S\times S}\right)\left(\bbeta^{*}_{S}-\bhbeta_{0,S}\right)}_{\infty}\le C\|\bSigma_{S\times S}^{-1}\|_{\infty}a_{n}\sqrt{(s+\log n)/m},
	\]
	and
	\begin{eqnarray*}
		\Abs{\bSigma_{S\times S}^{-1}\left(\z_{n}-\bhSigma\bbeta^{*}\right)_{S}}_{\infty}=O_{\pr}\left(\norm{\bSigma_{S\times S}^{-1}}_{\infty}\sqrt{\frac{\log n}{n}}+\norm{\bSigma_{S\times S}^{-1}}_{\infty}a_{n}^{2}\right).
	\end{eqnarray*}
	With Lemma \ref{prop:Bnbeta0} and the choice of $\lambda_{n}$, we obtain that 
	\[
	\Abs{	\btbeta_{S}-\bbeta^{*}_{S}
	}_{\infty}\le C\norm{\bSigma_{S\times S}^{-1}}_{\infty}\left(\sqrt{\frac{\log n}{n}}+a_{n}\sqrt{\frac{s\log n}{m}}\right).
	\]
	Then Theorem \ref{thm:support} (ii) follows from the above and together with the lower bound condition on $\min_{j\in S}\abs{\beta^{*}_{j}}$.
	
	Theorem \ref{thm:supportt} (ii) follows from the similar proof of Theorem \ref{thm:support} (ii) by replacing the initial estimator as $\bhbeta^{(t-1)}$ and the lower bound condition on $\min_{j\in S}\abs{\beta^{*}_{j}}$.
\end{proof}
\section{Additional Experiments}\label{sec:simsupp}
In this section we provide some additional experiment results using quantile level $\tau=0.5$. The results are reported in Tables \ref{mn_normal0.5}, \ref{mn_cauchy0.5} and \ref{mn_exp0.5}. The observations are similar to the case of $\tau=0.3$ in Section \ref{sec:effect}.
\begin{table}
	\caption{The $\ell_2$-error, precision, and recall of the three estimators under different combinations of the sample size $n$ and local sample size $m$. Noises are generated from normal distribution and quantile level $\tau = 0.5$.\label{mn_normal0.5}}
	\centering
	\resizebox{\textwidth}{!}{%
		\begin{tabular}{c|c|ccc|ccc|ccc}
			\hline
			\multicolumn{2}{c|}{$m$} & \multicolumn{3}{c|}{200} & \multicolumn{3}{c|}{500} & \multicolumn{3}{c}{1000} \\ \hline
			\multicolumn{2}{c|}{$n$} & 5000 & 10000 & 20000 & 5000 & 10000 & 20000 & 5000 & 10000 & 20000 \\ \hline
			
			\multirow{3}{*}{\begin{tabular}[c]{@{}c@{}}Pooled\\ REL\end{tabular}} & Precision      & 0.83  & 0.91  & 0.94  & 0.81  & 0.87  & 0.94  & 0.82  & 0.86  & 0.93  \\
			&Recall         & 1.00  & 1.00  & 1.00  & 1.00  & 1.00  & 1.00  & 1.00  & 1.00  & 1.00  \\
			&$\ell_2$-error & 0.133 & 0.097 & 0.071 & 0.131 & 0.094 & 0.069 & 0.130 & 0.098 & 0.070 \\ \hline
			\multirow{3}{*}{\begin{tabular}[c]{@{}c@{}}Dist\\ REL\end{tabular}} & Precision      & 0.98  & 0.99  & 1.00  & 0.91  & 0.95  & 0.98  & 0.86  & 0.90  & 0.96  \\
			&Recall         & 1.00  & 1.00  & 1.00  & 1.00  & 1.00  & 1.00  & 1.00  & 1.00  & 1.00  \\
			&$\ell_2$-error & 0.149 & 0.112 & 0.088 & 0.137 & 0.098 & 0.073 & 0.132 & 0.100 & 0.072 \\ \hline
			\multirow{3}{*}{\begin{tabular}[c]{@{}c@{}}Avg\\ DC\end{tabular}} & Precision      & 0.05  & 0.04  & 0.04  & 0.07  & 0.05  & 0.04  & 0.14  & 0.08  & 0.05  \\
			&Recall         & 0.99  & 1.00  & 1.00  & 0.98  & 0.99  & 1.00  & 0.97  & 0.99  & 0.99  \\
			&$\ell_2$-error & 0.341 & 0.324 & 0.313 & 0.219 & 0.202 & 0.192 & 0.174 & 0.156 & 0.139 \\ \hline
		\end{tabular}
	}
\end{table}
\begin{table}
	\caption{The $\ell_2$-error, precision, and recall of the three estimators under different combinations of the sample size $n$ and local sample size $m$. Noises are generated from Cauchy distribution and quantile level $\tau = 0.5$.\label{mn_cauchy0.5}}
	\centering
	\resizebox{\textwidth}{!}{%
		\begin{tabular}{c|c|ccc|ccc|ccc}
			\hline
			\multicolumn{2}{c|}{$m$} & \multicolumn{3}{c|}{200} & \multicolumn{3}{c|}{500} & \multicolumn{3}{c}{1000} \\ \hline
			\multicolumn{2}{c|}{$n$} & 5000 & 10000 & 20000 & 5000 & 10000 & 20000 & 5000 & 10000 & 20000 \\ \hline
			\multirow{3}{*}{\begin{tabular}[c]{@{}c@{}}Pooled\\ REL\end{tabular}} & Precision      & 0.82  & 0.88  & 0.94  & 0.85  & 0.91  & 0.95  & 0.83  & 0.89  & 0.93  \\
			&Recall         & 1.00  & 1.00  & 1.00  & 1.00  & 1.00  & 1.00  & 1.00  & 1.00  & 1.00  \\
			&$\ell_2$-error & 0.118 & 0.083 & 0.063 & 0.120 & 0.087 & 0.063 & 0.119 & 0.085 & 0.062 \\ \hline
			\multirow{3}{*}{\begin{tabular}[c]{@{}c@{}}Dist\\ REL\end{tabular}} & Precision      & 0.99  & 0.99  & 1.00  & 0.93  & 0.96  & 0.99  & 0.85  & 0.92  & 0.96  \\
			&Recall         & 1.00  & 1.00  & 1.00  & 1.00  & 1.00  & 1.00  & 1.00  & 1.00  & 1.00  \\
			&$\ell_2$-error & 0.135 & 0.098 & 0.072 & 0.124 & 0.091 & 0.066 & 0.121 & 0.087 & 0.064 \\ \hline
			\multirow{3}{*}{\begin{tabular}[c]{@{}c@{}}Avg\\ DC\end{tabular}} & Precision      & 0.05  & 0.04  & 0.04  & 0.09  & 0.06  & 0.04  & 0.14  & 0.08  & 0.06  \\
			&Recall         & 1.00  & 1.00  & 1.00  & 1.00  & 1.00  & 1.00  & 1.00  & 1.00  & 1.00  \\
			&$\ell_2$-error & 0.303 & 0.302 & 0.280 & 0.196 & 0.176 & 0.166 & 0.159 & 0.133 & 0.122 \\ \hline
		\end{tabular}
	}
\end{table}
\begin{table}
	\caption{The $\ell_2$-error, precision, and recall of the three estimators under different combinations of the sample size $n$ and local sample size $m$. Noises are generated from exponential distribution and quantile level $\tau = 0.5$.\label{mn_exp0.5}}
	\centering
	\resizebox{\textwidth}{!}{%
		\begin{tabular}{c|c|ccc|ccc|ccc}
			\hline
			\multicolumn{2}{c|}{$m$} & \multicolumn{3}{c|}{200} & \multicolumn{3}{c|}{500} & \multicolumn{3}{c}{1000} \\ \hline
			\multicolumn{2}{c|}{$n$} & 5000 & 10000 & 20000 & 5000 & 10000 & 20000 & 5000 & 10000 & 20000 \\ \hline
			\multirow{3}{*}{\begin{tabular}[c]{@{}c@{}}Pooled\\ REL\end{tabular}} & Precision      & 0.83  & 0.87  & 0.98  & 0.82  & 0.93  & 0.95  & 0.83  & 0.92  & 0.95  \\
			&Recall         & 1.00  & 1.00  & 1.00  & 1.00  & 1.00  & 1.00  & 1.00  & 1.00  & 1.00  \\
			&$\ell_2$-error & 0.071 & 0.053 & 0.035 & 0.068 & 0.050 & 0.038 & 0.065 & 0.047 & 0.034 \\ \hline
			\multirow{3}{*}{\begin{tabular}[c]{@{}c@{}}Dist\\ REL\end{tabular}} & Precision      & 0.89  & 0.97  & 0.99  & 0.86  & 0.95  & 0.98  & 0.85  & 0.94  & 0.99  \\
			&Recall         & 1.00  & 1.00  & 1.00  & 1.00  & 1.00  & 1.00  & 1.00  & 1.00  & 1.00  \\
			&$\ell_2$-error & 0.086 & 0.063 & 0.040 & 0.073 & 0.051 & 0.039 & 0.072 & 0.049 & 0.035 \\ \hline
			\multirow{3}{*}{\begin{tabular}[c]{@{}c@{}}Avg\\ DC\end{tabular}} & Precision      & 0.08  & 0.07  & 0.06  & 0.07  & 0.08  & 0.05  & 0.09  & 0.06  & 0.07  \\
			&Recall         & 1.00  & 1.00  & 1.00  & 1.00  & 1.00  & 1.00  & 1.00  & 1.00  & 1.00  \\
			&$\ell_2$-error & 0.188 & 0.188 & 0.179 & 0.100 & 0.098 & 0.095 & 0.085 & 0.073 & 0.063 \\ \hline
			
		\end{tabular}
	}
\end{table}
\newpage
\bibliographystyle{rss}
\bibliography{ref}

\end{document}